\documentclass[english,11pt]{article}
\usepackage[margin=1.0in]{geometry}
\usepackage[numbers]{natbib}
\usepackage{algorithm}
\usepackage{algorithmic}

\usepackage[pdftex,colorlinks]{hyperref}
\usepackage{float}
\usepackage{amssymb}
\usepackage{amsmath}
\usepackage{amsfonts}
\usepackage{amssymb}
\usepackage{amsthm}
\usepackage{algorithm}
\usepackage{algorithmic}
\newtheorem{theorem}{Theorem}[section]
\newtheorem{lemma}[theorem]{Lemma}

\newtheorem{counter-example}[theorem]{Counter example}

\newtheorem{open question}[theorem]{Open question}
\newtheorem{corollary}[theorem]{Corollary}

\newtheorem{question}[theorem]{Question}
\newtheorem{definition}[theorem]{Definition}

\newtheorem{claim}{Claim}

\newcommand{\ch}{{\cal H}}

\newcommand{\cf}{{\cal F}}

\newcommand{\cy}{{\cal Y}}

\newcommand{\cs}{{\cal S}}

\DeclareMathOperator*{\Dim}{Dim}

\makeatletter

\floatstyle{ruled}
\newfloat{algorithm}{tbp}{loa}
\providecommand{\algorithmname}{Algorithm}
\floatname{algorithm}{\protect\algorithmname}

\makeatother

\usepackage{babel}

\begin{document}
\author{Amit Daniely\thanks{Dept. of Mathematics, The Hebrew University, Jerusalem, Israel}   \hspace{1cm} Michael Schapira\thanks{School of Computer Science and Engineering, The Hebrew University, Jerusalem, Israel.} \hspace{1cm}Gal Shahaf\thanks{Dept. of Mathematics, The Hebrew University, Jerusalem, Israel}
}

\setcounter{page}{0}

\title{Inapproximability of Truthful Mechanisms via Generalizations of the VC Dimension}

\maketitle
\thispagestyle{empty}
\begin{abstract}
Algorithmic mechanism design (AMD) studies the delicate interplay between computational efficiency, truthfulness, and optimality. We focus on AMD's paradigmatic problem: combinatorial auctions. We present a new generalization of the VC dimension to multivalued collections of functions, which encompasses the classical VC dimension, Natarajan dimension, and Steele dimension. We present a corresponding generalization of the Sauer-Shelah Lemma and harness this VC machinery to establish inapproximability results for deterministic truthful mechanisms. Our results essentially unify all inapproximability results for deterministic truthful mechanisms for combinatorial auctions to date and establish new separation gaps between truthful and non-truthful algorithms.
\end{abstract}

\newpage

\section{Introduction}

Algorithmic mechanism design (AMD) studies computational environments in which the input to the algorithm is provided by self-interested, strategic, parties, e.g., combinatorial auctions---the now paradigmatic problem of AMD: $m$ items $1,\ldots,m$ are being sold to $n$ bidders $1,\ldots,n$. Each bidder $i$ has a valuation function $v_i : 2^{[m]}\rightarrow R$, which specifies $i$'s ``maximum willingness to pay'' for every subset (``bundle'') of items $S\subseteq [m]$. The objective is to maximize social welfare, that is, to partition the items between the bidders in such a way that the ``social welfare'' $\Sigma_i v_i(S_i)$ is maximized, where $S_i$ is the bundle assigned to bidder $i$. A prominent line of research in AMD is exploring the complex interplay between three natural desiderata: (1) computational efficiency; (2) truthfulness, i.e., incentivizing bidders to reveal their actual valuations to the mechanism; and (3) approximation guarantees.

Past studies expose inherent tensions between these desiderata and establish large gaps between truthful and non-truthful mechanisms in various environments, including combinatorial auctions~\cite{dobzinski2011impossibility,dobzinski2012computational,papadimitriou2008hardness}---the paradigmatic setting of AMD. Yet, despite much effort along these lines, long-standing questions remain wide open, including the ``holy grail''~\cite{DobzinskiNS12}: designing a \emph{deterministic} truthful mechanism for general combinatorial auctions that matches the approximation guarantee of the best non-truthful algorithm (or prove that no such mechanism exists). Beyond the contribution of research along these lines to AMD, it also yielded ideas and insights that are of broader interest. One example is the interesting connection between the VC dimension and optimization over partial domains, presented in~\cite{papadimitriou2008hardness}. Papadimitriou et al.~\cite{papadimitriou2008hardness} showed how classical VC machinery (namely, the Sauer-Shelah Lemma~\cite{sauer1972density,shelah1972combinatorial}) can be used to prove inapproxiability results for deterministic truthful mechanisms in the context of combinatorial public projects. Subsequent work further developed this idea in the combinatorial public projects setting~\cite{papadimitriou2008hardness,buchfuhrer2010computation} and in the combinatorial auctions setting~\cite{dughmi2009amplified,buchfuhrer2009limits, buchfuhrer2010inapproximability,mossel2009vc}.


While combinatorial public projects
naturally lend themselves to standard VC lower bounding techniques, other environments pose a more a complex challenge. \cite{dughmi2009amplified,buchfuhrer2009limits, buchfuhrer2010inapproximability,mossel2009vc} showed that sophisticated adaptations of existing VC machinery, namely the Sauer-Shelah Lemma, can yield inapproximability results for combinatorial auctions in a specific setting: VCG-based mechanisms for ``capped-additive valuations''. However, as combinatorial auctions deal with \emph{partitions} of a universe (as opposed to subsets), it is not clear how standard VC machinery can, in general, be applied to other auction contexts (general truthful mechanisms, other types of valuations, communication complexity,\ldots).

We present a new generalization of the VC dimenstion that is both natural from a combinatorial perspective, and encompasses the classical VC dimension, as well as past generalizations of the VC dimension: the Natarajan dimension~\cite{Natarajan89b} and the Steele dimension~\cite{steele1978existence}.  While many previously proposed generalizations of the VC dimension~\cite{sauer1972density,shelah1972combinatorial,Natarajan89b,haussler1995generalization,Ben-DavidCeHaLo95,AlonBeCsHa97,DanielySaBeSh11,DanielySh14} are motivated by machine learning applications---proving positive results for classification---our generalization is aimed at establishing negative results in AMD, and is thus of a very different flavour. We prove a corresponding generalization of the Sauer-Shelah Lemma that generalizes several known bounds~\cite{sauer1972density,shelah1972combinatorial, Natarajan89b,haussler1995generalization, steele1978existence} and apply it to obtain new inapproximability results for deterministic truthful mechanisms and also for simplifying and unifying existing inapproximability results.

\subsection{Generalizing the VC dimension and Implications for Auctions}\label{subsec:illustration}

Recall the classic Vapnik-Chervonenkis dimension \cite{Vapnik95}.

\begin{definition}
[\bf VC-dimension] Let $\mathcal{H}\subset \{0,1\}^X$. A subset $S\subset X$ is \textbf{shattered} by $\mathcal{H}$ if $\mathcal{H}|_S=\{0,1\}^S$. The \textbf{VC-dimension} of $\ch$, denoted $VC(\mathcal{H})$, is the maximal cardinality
of a subset $S\subset X$ that is shattered by $\mathcal{H}$.
\end{definition}

The prominence of the $VC$ dimension is largely due to the Sauer-Shelah Lemma \cite{sauer1972density,shelah1972combinatorial}:

\begin{lemma}
[\bf Sauer-Shelah] For every $\mathcal{H}\subset \{0,1\}^X$,
$|\ch |\le\sum_{i=1}^{VC(\ch)}\binom{|X|}{i}$
\end{lemma}

Our generalization of the VC dimension relies on the notion of ``shattering to $k$ values'':

\begin{definition}[\bf $k$-shattering]
Let $\ch\subset Y^X$ and $k\ge 2$. We say that $A\subset X$ is {\bf $k$-shattered} by $\ch$, if for every $a\in A$ there is a set $Y_a\subset Y$ of size $k$ such that the following holds: Every function $f:A\rightarrow Y$ that satisfies $f(a)\in Y_a$ for all $a\in A$ is a restriction of some function in $\ch$. That is, $\exists h\in \ch$ s.t. $f=h|_A$.
\end{definition}

Putting it differently, for every choice of elements $\{y_a\}_{a\in A}$ such that $y_a\in Y_a,$ there is some $h\in \ch$ such that $h(a)=y_a$ for all $a\in A$.

\begin{definition}[\bf $k$-dimension]
We define the {\bf $k$-dimension} of $\ch$, denoted $\Dim_k(\ch)$, as the maximal cardinality of a $k$-shattered set.
\end{definition}

This notion of (k-)shattering differs from past generalizations of shattering to non-binary domains (see, e.g., \cite{buchfuhrer2009limits,dughmi2009amplified}). We present the following generalization of the Sauer-Shelah Lemma:

\begin{theorem}\label{thm:sauer_gen}
For every $\ch\subset Y^X$ and for every $2\le k\le |Y|$,
\[
|\ch|\le \sum_{i=0}^{\Dim_k(\ch)} \binom{|X|}{i} (k-1)^{|X|-i}\binom{|Y|}{k}^{i}\le |X|^{\Dim_k(\ch)}|Y|^{k\Dim_k(\ch)}(k-1)^{|X|}
\]
\end{theorem}

We point out that (1) taking $k=|Y|=2$, we get the $VC$ dimension and the Sauer-Shelah Lemma~\cite{sauer1972density,shelah1972combinatorial}; (2) for general $Y$ and $k=2$ we get Natarajan's dimension \cite{Natarajan89b} and the bound in \cite{haussler1995generalization}, which strengthen Natarajan's bound~\cite{Natarajan89b}; and (3) for $k=|Y|$, we get Steele's dimension and bound~\cite{steele1978existence}. Our proof of Theorem \ref{thm:sauer_gen} (the generalized Sauer-Shelah Lemma) relies on a careful application of techniques used for proving previous bounds (e.g., \cite{steele1978existence,Natarajan89b,haussler1995generalization}) and appears in Appendix~\ref{apx_sauer_gen}.

We now provide an intuitive exposition of the connection between our dimension/bound and combinatorial auctions. We focus, for ease of exposition, on the class of VCG-based, a.k.a. maximal-in-range (MIR) mechanisms, and on bidders with particularly simple ``single-minded'' valuations: Each bidder $i$ is only interested in a single bidder-specific bundle of items $T_i$, and assigns a value of $1$ to all sets of items containing $T_i$ and a value of $0$ to all other bundles of items. A maximal-in-range mechanism has a fixed, predetermined, ``bank of allocations'' of the items to the bidders, and for every input (reported valuations) outputs the best allocation in this bank. A trivial MIR mechanism is the mechanism that always allocates all items to the bidder who values them the most, i.e., the MIR mechanism whose bank of allocations consists of the $n$ partitions of the items of the form $(\emptyset,\ldots,\emptyset,[m],\emptyset,\ldots,\emptyset)$. A simple argument shows that the approximation ratio of this mechanism is $\min\{n,m\}$. We prove that this naive mechanism is essentially the best VCG-based (MIR) mechanism: no VCG-based mechanism has approximation ratio $m^{1-\epsilon}$ for any constant $\epsilon>0$ unless NP$\subseteq$P/poly.

Suppose, for point of contradiction, that there exists a computationally-efficient MIR mechanism $M$ with bank of allocations $\ch$ that has approximation ratio $m^{1-\epsilon}$. Observe that the allocations in $\ch$ can naturally be regarded as a collection of functions from $[m]$ to $[n]\cup\{*\}$; every allocation $A\in\ch$ is associated with function $f_A$ that maps every item in $[m]$ to a bidder in $[n]$, or leaves the item unallocated (mapping it to $*$), as in $A$. We now prove, via a subtle argument that utilizes Theorem~\ref{thm:sauer_gen}, that there exists a large (polynomial in $m$) set of bidders $X$ and a large (polynomial in $m$) set of items $Y$, such that all partitions of the items in $Y$ between bidders in $X$ appear in $\ch$ and, in this sense, $\ch$ shatters the pair $(X,Y)$. We can now conclude that, as $M$ optimizes exactly over the allocations in $\ch$, $M$ can compute the optimal allocation of $|X|$ items to $|Y|$ single-minded bidders---an NP-hard task (this is, in fact, precisely the classical SET-PACKING problem).

While simple, this example illustrates the strength of our VC machinery with respect to past utilizations of VC dimension arguments in the context of combinatorial auctions: (1) $\ch$ shatters a polynomial number of bidders and a polynomial number of items; and (2) every partition in which \emph{all} items in $X$ are assigned to bidders in $Y$ appears in $\ch$. Observe that shattering a constant number of bidders (e.g., $2$ bidders in~\cite{buchfuhrer2010inapproximability}) is insufficient, as combinatorial auctions with single-minded bidders are, in fact, tractable for a constant number of bidders. Also, if even a single allocation of all items in $X$ to bidders in $Y$ does not appear in $\ch$, the reduction from a combinatorial auction with $|X|$ items and $|Y|$ single-minded bidders is no longer possible (as the optimal allocation might simply not be in $M$'s bank).

\subsection{Inapproximability Results}

Our inapproximability results for combinatorial auctions are summarized in Table~\ref{table}. Our results are categorized in three dimensions:

\begin{table}[h]
\caption{Inapproximability results for truthful mechanisms}\label{table}
\begin{tabular}{|l|ll|ll|ll|ll|}
\hline
\textbf{}         & \multicolumn{2}{c|}{\textbf{Comp.}}                     & \multicolumn{2}{c|}{\textbf{Value}}                     & \multicolumn{2}{c|}{\textbf{Demand}}      & \multicolumn{2}{c|}{\textbf{Comm.}}       \\
\textbf{}         & \textbf{All}               & \textbf{VCG}               & \textbf{All}               & \textbf{VCG}               & \textbf{All} & \textbf{VCG}               & \textbf{All} & \textbf{VCG}               \\ \hline
\textbf{General}  & $m^{1-\epsilon}$           & $m^{1-\epsilon}$           & -                          & -                          & ?            & $m^{1-\epsilon}$           & ?            & $m^{1-\epsilon}$           \\
                  & {[}new{]}                  & {[}new{]}                  &                            &                            &              & {[}new{]}                  &              & {[}new{]}                  \\ \hline
\textbf{Single-}  & -                          & $m^{1-\epsilon}$           & -                          & -                          & -            & -                          & -            & -                          \\
\textbf{Minded}   &                            & {[}new{]}                  &                            &                            &              &                            &              &                            \\ \hline
\textbf{Submodular}     & $m^{\frac{1}{2}-\epsilon}$ & $m^{\frac{1}{2}-\epsilon}$ & $m^{\frac{1}{2}-\epsilon}$ & $m^{\frac{1}{2}-\epsilon}$ & ?            & $m^{\frac{1}{3}-\epsilon}$ & ?            & $m^{\frac{1}{3}-\epsilon}$ \\
                  & {[}revisit \citealp{dobzinski2012computational}{]}              & {[}rev. \citealp{dobzinski2012computational,buchfuhrer2010inapproximability}{]}              & {[}rev. \citealp{dobzinski2011impossibility}{]}               & {[}rev. \citealp{dobzinski2011impossibility}{]}               &              & {[}new{]}                  &              & {[}new{]}                  \\ \hline
\textbf{Capped-}  & ?                          & $m^{\frac{1}{2}-\epsilon}$ & -                          & -                          & -            & -                          & -            & -                          \\
\textbf{Additive} &                            & {[}rev. \citealp{buchfuhrer2010inapproximability}{]}            &                            &                            &              &                            &              &                            \\ \hline
\textbf{k-Dup}    & ?                          & $m^{1-\epsilon}$           & -                          & -                          & ?            & $m^{1-\epsilon}$           & ?            & $m^{1-\epsilon}$           \\
                  &                            & {[}new{]}                  &                            &                            &              & {[}new{]}                  &              & {[}new{]}                  \\ \hline
\end{tabular}
\\
{\em The table only includes lower bounds that rely on the truthfulness of the mechanisms. All results in the table are asymptotically tight, except the $m^{\frac{1}{3}-\epsilon}$ bound for VCG-based mechanisms w.r.t. submodular valuations. Question marks (``?'') indicate that there are no separation results between truthful and non-truthful mechanisms. Minus signs (``-'') indicate that inapproximability results for truthful mechanisms are irrelevant, either because of irrelevance of the model to the class of valuations (e.g., communication complexity when valuations are succinctly described), or because lower bounds for non-truthful algorithms match the existing upper bounds for truthful mechanisms.}
\end{table}

\begin{itemize}

\item {\bf Representation of the ``input''.} We consider the two standard models for accessing the ``input''---the valuation functions---in combinatorial auctions: (1) the computational complexity model, in which valuation functions are succinctly encoded, and mechanisms must run in time that is polynomial in the input length; and (2) the oracle model, in which valuations are treated as black boxes that can only answer a certain type of queries, and complexity is measured in terms of the number of queries. Three types of queries are commonly considered: (i) value queries; (ii) demand queries; and (iii) the communication complexity model, in which oracles can answer any type of query (addressed to a single valuation).


\item {\bf General deterministic mechanisms vs. VCG-based mechanisms.} The Vickrey-Clarke-Groves (VCG) scheme for designing truthful mechanisms is the only universal technique for designing truthful deterministic mechanisms. While a naive application of VCG is often computationally intractable, more clever uses of the VCG scheme provide the best deterministic truthful approximation algorithms for combinatorial auctions to date~\cite{dobzinski2005approximation,Holzman2004104}. We present inapproximability results for both general (unrestricted) deterministic mechanisms and for the important subcategory of VCG-based mechanisms

\item{\bf Classes of valuation functions.} Over the past decade, much effort has been invested in bounding the approximability guarantees of truthful mechanisms for different classes of valuation functions of interest, including general valuations~\cite{Holzman2004104}, submodular valuations~\cite{dobzinski2011impossibility,dobzinski2012computational}, single-minded valuations~\cite{lehmann2002truth}, capped additive valuations~\cite{buchfuhrer2010inapproximability}, and more. We present inapproximability results for several well-studied classes of valuation functions.
\end{itemize}

We shall now briefly highlight some of the new results in Table~\ref{table}:

\vspace{0.05in}\noindent{\bf Inapproximability for general deterministic mechanisms.} We prove that no computationally-efficient and truthful mechanism for general valuations can (asymptotically) outperform the trivial $m$-approximation mechanism that bundles all items together and assigns them to the bidder who values them most (in a 2nd-price auction). Specifically, for any choice of constant $\epsilon>0$, there exists a class of succinctly-described valuation functions such that (1) a non-truthful computationally-efficient algorithm achieves an approximation ratio of $m^\epsilon$ for this class ; but (2) no computationally-efficient and truthful algorithm (mechanism) can obtain an $O(m^{1-\epsilon})$-approximation unless NP $\subseteq P/poly$. Our proof of this separation gap between truthful and non-truthful algorithms relies on the ``direct hardness approach''~\cite{dobzinski2011impossibility} combined with a careful application of the Sauer-Shelah Lemma.

\vspace{0.05in}\noindent{\bf Results for VCG-based mechanisms.} We present several new inapproximability results for VCG-based mechanisms. Our proofs of these results rely on lower bounding the Steele dimension (yet another special case of our generalization). Our results establish, in particular (1) that the best deterministic approximation ratio for general valuations is $\frac{m}{\sqrt{\log(m)}}$, obtained via a simple VCG mechanism~\cite{Holzman2004104}, is essentially tight; (2) a $m^{\frac{1}{3}}$ lower bound for submodular valuations in the communication complexity model, improving upon the $m^{\frac{1}{6}}$-inapproximability result of Dobzinski and Nisan~\cite{dobzinski2007limitations}, obtained via a short and elegant proof that utilizes the Natarajan dimension (another special case of our generalized VC dimension); (3) that no VCG-based mechanism can match the approximation guarantees of truthful non-VCG mechanisms in combinatorial auctions with single-minded bidders~\cite{lehmann2002truth} and in combinatorial auctions with multiple duplicates of each item~\cite{bartal2003incentive}.\vspace{0.05in}

\noindent{\bf Unifying all inapproximability results for deterministic mechanisms.} We show that essentially all inapproximability results for deterministic truthful mechanisms in the combinatorial auctions setting (to date) can be proved in our generalized VC dimension framework. Inapproximability results that fit in our framework include the inapproximability results for (1) submodular valuations in the value queries model~\cite{dobzinski2011impossibility}; (2) submodular valuations in the computational complexity model~\cite{dobzinski2012computational}; and (3) capped additive valuations~\cite{buchfuhrer2010inapproximability}. We believe that, in this sense, our new approach marks the borderline of the state of the art and that further progress along these lines must thus entail inherently new ideas.

\subsection{Organization}

We discuss the connections between Theorem~\ref{thm:sauer_gen} and partitions in Section~\ref{sec:partitions}. We then explain how our techniques can be applied to prove inapproximability results for VCG-based and unrestricted truthful deterministic mechanisms in Section~\ref{sec:vcg-techniques} and Section~\ref{sec:general-techniques}, resprectively. See Appendix for the proof of Theorem~\ref{thm:sauer_gen}, background on mechanism design and combinatorial auctions (Section~\ref{sec:background}) and a detailed exposition of our inapproximability results for truthful mechanisms (Sections \ref{sec:results-first}-\ref{sec:results-last}). 

\section{Shattering vs. Approximation}\label{sec:partitions}

We now discuss the connections between Theorem \ref{thm:sauer_gen} (the generalized Sauer-Shelah Lemma) and so called ``approximate-partitions''. We consider collections $\ch$ of allocations of an item set $X$ to a set of indices $Y$. We will show that if $\ch$ (even very loosely) in some sense approximates the collection of all partitions, then there are large subsets $S\subset X$ and $A\subset Y$ for which {\em all} the partitions of items in $S$ to indices in $A$ belong to $\ch$.

We first define the notions of shattering by collections of allocations (and in particular we accommodate scenarios that involve duplicate elements). Then, we define two approximation notions for collections of allocations, and prove shattering results.

\subsection{Shattering Allocations}

We extend the notion of shattering to allocations, which are more natural than functions in the context of combinatorial auctions.
Let $X$ be a set of items and $Y$ a set of indices. An {\em allocation} is a pairwise disjoint collection $\{S_y\}_{y\in Y}$ of subsets of $X$. If $\cup_{y\in Y}S_y=X$, we say that the allocation is a {\em partition}.
We denote the collection of all allocations of $X$ to the indices $Y$ by $P(X,Y)$.
The collection of allocations naturally corresponds to the collection of functions $f:X\to Y\cup \{*\}$ (here, $f^{-1}(*)$ is the set of items that were not allocated to any index). We will freely alternate between these two representations. We say that a collection $\ch$ of allocations {\em shatters} a pair of sets $S\subset X$ and $A\subset Y$ if all partitions of $S$ to indices in $A$ are induced by allocations from $\ch$. Namely,
\begin{definition}
Let $\ch\subset \left(Y\cup\{*\}\right)^{X}$ be a collection of allocations. We say that a pair of subsets $S\subset X$ and $A\subset Y$ is {\em shattered} if $A^{S}\subset \ch|_{S}$.
\end{definition}
We wish to accommodate scenarios in which each item have $d$ (identical) copies. To this end we define:
\begin{definition}
Let $d\ge 1$. A \textbf{d-duplicate allocation} of a set $X$ to indices $Y$ is a collection $\{S_y\}_{y\in Y}$ of subsets of $X$ such that every $x\in X$ belongs to $\le d$ subsets from $\{S_y\}_{y\in Y}$.
\end{definition}
Note that $1$-duplicate allocation is just a simple allocation. We denote by $P_d(X,Y)$ the set of all $d$-duplicate allocations. As with standard allocations, we say that a collection $\ch$ of $d$-duplicate allocations {\em shatters} a pair of sets $S\subset X$ and $A\subset Y$ if every partition of $S$ to the indices $A$ is induced by $\ch$. Namely, we let $\mathcal{H}|_{S,A}:=\{f:S\to A|\exists \{S_y\}_{y\in Y}\in \mathcal{H} \ \forall y\in A,\;f^{-1}(y)=S_y\cap S\}$,
and define:
\begin{definition} \textbf{(Shattering of allocations)}
Let $\ch$ be a collection of $d$-duplicate allocations of a set $X$ to indices $Y$.
A pair of subsets $S\subset X$, $A\subset Y$ is {\bf shattered} by $\mathcal{H}$, if $A^{S}=\mathcal{H}|_{S,A}$.
\end{definition}
We will also use the notion of shattering of a single index.
\begin{definition}
Let $\ch$ be a collection of $d$-duplicate allocations of a set $X$ to indices $Y$.
A pair $S\subset X$, $a\subset Y$ is {\bf shattered} by $\mathcal{H}$, if for every $T\subset S$ there is $f\in\mathcal{H}|_{S,Y}$ with $T=f^{-1}(a)$.
\end{definition}
Note that if a pair $S\subset X,A\subset Y$ is shattered then, for every $a\in A$, the pair $S,a$ is shattered.

\subsection{Approximate Containment}
Let $X$ be a set of $m$ items and $Y$ a set of $n$ indices.
\begin{definition}
Let $\ch\subset P_d(X,Y)$ be a collection of allocations and let $\alpha\ge 1$. We say that $\ch$ has the {\bf $\alpha$-containment property} if for every allocation $\{T_y\}_{y\in Y}$ of $X$, there is an allocation $\{S_y\}_{y\in Y}\in \ch$ such that $\frac{1}{\alpha}$ percent of the non-empty sets in $\{T_y\}_{y\in Y}$ are covered by the corresponding set in $\{S_y\}_{y\in Y}$. Namely,
$|\{y\mid \emptyset \ne T_y\subset S_y\}|\ge \frac{1}{\alpha}|\{y\mid \emptyset \ne T_y\}|$
\end{definition}

\begin{theorem}\label{thm:shatter_single_minded_0}
Let  $\epsilon>0$, $\mathcal{H}\subset P_d(X,Y)$ and assume that $n\geq m^{1-3\epsilon/4}$
and $\mathcal{H}$ has the $m^{1-\epsilon}$-containment property. There exists a shattered pair $S\subset X$, $A\subset Y$ of sizes  $\tilde{\Omega}(m^{3\epsilon/4})$ and $m^{\epsilon/4}$.
\end{theorem}

\begin{proof}
For the sake of simplicity, let us restrict our attention to the case that the number of indices, $n$, is $m^{1-\frac{3\epsilon}{4}}$. We assume also that $m^{\frac{3\epsilon}{4}}=\frac{m}{n}$ and $k:=m^{\frac{\epsilon}{4}}$ are integers.

Consider a random partition, $\{T_y\}_{y\in Y}$, of $X$ that gives exactly $m^{\frac{3\epsilon}{4}}$ items to every index $y\in Y$.  Since $\ch$ has the $m^{1-\epsilon}$-containment property, at least $k$ sets in $\{T_y\}_{y\in Y}$ must be covered by some allocation in $\ch$. Namely, there is some $\{S_y\}_{y\in Y}\in\ch$ such that $T_y\subset S_y$ for at least $k$ indices in $Y$.

In addition to the random partition, consider now a subset $A\subset Y$ of $k$ indices, chosen uniformly at random and independently from $\{T_y\}_{y\in Y}$. Since $k$ indices must be covered (i.e. there is an allocation in $\ch$ such that the sets corresponding to these indices contain the sets sampled to these indices), the probability that all indices in $A$ (over the choice of both $A$ and $\{T_y\}_{y\in Y}$) are covered is $\ge \frac{1}{\binom{n}{k}}\ge \frac{1}{\binom{m}{k}}\ge m^{-k}$.
Hence, there exists a (fixed) set $A$ of $k$ indices which are covered w.p. $\ge m^{-k}$ when their corresponding sets are sampled at random (by the above distribution).

By conditioning on the set which is the union of the sets corresponding to the indices in $A$, it follows that there exists a set $S\subset X$ with $|S|=k\cdot m^{\frac{3\epsilon}{4}}=m^{\epsilon}$ such that if a function $f:S\to A$ is sampled uniformly at random among all functions with $\forall i,j\in A,\;\;|f^{-1}(i)|=|f^{-1}(j)|$, then w.p. $\ge m^{-k}$, it holds that $f\in \ch_{S,A}$.

It follows that $|\ch_{S,A}|\ge k^{m^{\epsilon}}\cdot m^{-k}\cdot m^{-\epsilon k}$ (the term $m^{-\epsilon k}$ is a lower bound on the probability that a uniformly chosen function from $S$ to $A$ will satisfy $\forall i,j\in A,\;\;|f^{-1}(i)|=|f^{-1}(j)|$).

By Theorem \ref{thm:sauer_gen}, be have that $|\ch_{S,A}|\le (k-1)^{m^{\epsilon}}m^{\epsilon \Dim_k(\ch_{S,A})}$. It follows that
\begin{eqnarray*}
m^{\epsilon \Dim_k(\ch_{S,A})} &\ge & \left(\frac{k}{k-1}\right)^{m^{\epsilon}}\cdot m^{-k}\cdot m^{-\epsilon k}
\\
&\ge &\left(\frac{k}{k-1}\right)^{m^{\epsilon}}\cdot m^{-2k}
\\
&=& e^{\ln\left(\frac{k}{k-1}\right)m^{\epsilon}-2k\ln(m)}
\\
&=& e^{\ln\left(1+\frac{1}{k-1}\right)m^{\epsilon}-2k\ln(m)}
\\
&\ge & e^{\frac{1}{2k} m^{\epsilon}-2k\ln(m)}
\end{eqnarray*}
The last inequality is correct for large $k$ since $\ln(1+x)=x+o(x)$. Taking logarithms we conclude that
$\Dim_k(\ch_{S,A})=\tilde{\Omega}(m^{\frac{3\epsilon}{4}})$.
\end{proof}

\subsection{Approximate Intersection}
Let $X$ be a set of $m$ items and $Y$ a set of $n$ indices.
\begin{definition}
Let $\ch\subset P_d(X,Y)$ be a collection of allocations and let $\alpha\ge 1$. We say that $\ch$ has the {\bf $\alpha$-intersection property} if for every allocation $\{T_y\}_{y\in Y}$ of $X$, there is an allocation $\{S_y\}_{y\in Y}\in \ch$ that agree with $\{T_y\}_{y\in Y}$ on $\frac{1}{\alpha}$ percent of the items. Namely,
\[
\sum_{y\in Y}|T_y\cap S_y|\ge \frac{1}{\alpha} \sum_{y\in Y}|T_y|
\]
\end{definition}
\begin{theorem}\label{thm:shatter_additive_0}
Suppose $n\ge m^{\frac{1}{k+1}-\epsilon}$ and $\ch\subset P(X,Y)$ has the $\frac{n}{1+2k}$-intersection property. There is a shattered pair $S\subset X,\;A\in\binom{Y}{k}$ with $|S|\ge \frac{m^{(k+1)\epsilon}}{\log_2(m)}$.
\end{theorem}
We will use the following lemma from \cite{buchfuhrer2009limits} (their lemma 5).

\begin{lemma}[\cite{buchfuhrer2009limits}]\label{lem:buch_umans}
Assume that  $\ch\subset P(X,Y)$ has the $\frac{n}{1+2k}$-intersection property. There exists a set $S\subset X$ of size $|S|\ge \frac{m}{n}$ such that $|\ch_{S,Y}|\ge (1+k)^{\frac{m}{n}}$
\end{lemma}

\begin{proof}(of theorem \ref{thm:shatter_additive_0})
For simplicity, we assume that $n= m^{\frac{1}{k+1}-\epsilon}$ (otherwise, we will look on the restriction of $\ch$ to a fixed set of $m^{\frac{1}{k+1}-\epsilon}$ indices).
By lemma \ref{lem:buch_umans}, there is a subset $U$ of size $\ge\frac{m}{n}$ such that $|\ch_{U,Y}|\ge (k+1)^{\frac{m}{n}}$. Applying theorem \ref{thm:sauer_gen} we get that
\begin{equation*}
(k+1)^{\frac{m}{n}} \le
m^{\Dim_k(\ch)}n^{k\Dim_k(\ch)}(k-1)^{\frac{m}{n}} \le m^{2\Dim_k(\ch)}(k-1)^{\frac{m}{n}}
\end{equation*}
Taking logarithm we get a $k$-shattered set $T\subset U$ of size $\frac{\log_2\left(\frac{k+1}{k-1}\right)\frac{m}{n}}{2\log_2(m)}\ge\frac{\frac{m}{n}}{2(k-1)\log_2(m)}$. Let $\{Y_a\}_{a\in T}$ be a collection of subsets of $Y$ that indicates that $T$ is $k$-shattered. Since there are at most $\binom{n}{k}\le \frac{n^k}{(k-1)2}$ possible options for each $Y_a$, by the pigeonhole principle, there is a subset $S\subset T$ of size $\frac{\frac{m}{n^{k+1}}}{\log_2(m)}\ge \frac{\frac{m}{m^{1-(k+1)\epsilon}}}{\log_2(m)}\ge
\frac{m^{(k+1)\epsilon}}{\log_2(m)}$ such that all the subsets $\{Y_a\}_{a\in S}$ are the same and equal to some $A\in\binom{Y}{k}$. This shows that the pair $(S,A)$ is shattered.
\end{proof}

\section{Inapproximability for VCG-Based Mechanisms}\label{sec:vcg-techniques}

We now show how our VC dimension arguments can be applied lower bound the approximability of VCG-based mechanisms. Consider the classical combinatorial auction setting: $m$ items $1,\ldots,m$, are sold to $n$ bidders $1,\ldots,n$, and each bidder $i$ has a private valuation function over bundles of items $v_i: 2^{[m]}\rightarrow R$. The objective is to partition the items between the bidders so as to maximize $\Sigma_i v_i(S_i)$, where $S_i$ is the bundle assigned to bidder $i$. A more general setting is that of a {\em combinatorial auction with duplicates}, in which $d$ identical units of each item are available but each bidder desires at most 1 unit of each item. The mechanism can now output a {\em $d$-duplicate} allocation, i.e., a collection $S=\{S_i\}_{i\in [n]}$ of subsets of $[m]$ such that each $x\in [m]$ belongs to at most $d$ of these subsets.

A VCG-based, aka \textit{maximal-in-range} ({=\em MIR}), mechanism has fixed {\em bank} of allocations $\ch\subset P_d([m],[n])$ such that the output of $M$ for every $n$-tuple of valuation function $v_1,\ldots,v_n$ is an allocation $S$ that maximizes $\sum_{i=1}^n v_i(S_i)$ over all allocations $S\in\ch$. Charging VCG prices from the bidders ensures that such mechanisms are truthful. 

We lower bound the approximability of MIR mechanisms in two steps: (1) using the machinery developed in the previous section to show that a ``good'' approximation ratio implies shattering of a large number of bidders and items by $\ch$; and (2) embedding a hard (e.g., NP-hardness, under the Unique Games Conjecture, and communication complexity) problem in this shattered domain.

\vspace{0.05in}\noindent{\bf Step I: Shattering a large number of bidders and items (the VC step).} Consider the following two of the simplest classes of valuation functions:

\begin{itemize}
\item $\mathbf{Single-Minded}$ (0/1): all valuations of the form $v(S)=1[T\subset S]$ for some $T\subset [m]$.
\item $\mathbf{Additive}$ (0/1): all valuations of the form $v(S)=w^T\chi(S)$, where $w\in\{0,1\}^n$.
\end{itemize}

Intuitively, these valuations correspond to the containment and the intersection properties defined in the previous section. Shattering is thus immediately guaranteed by the following corollaries of Theorem~\ref{thm:shatter_single_minded_0} and Theorem~\ref{thm:shatter_additive_0}, respectively:

\begin{corollary}\label{thm:shatter_single_minded}
Let  $\epsilon>0$, and let $\mathcal{H}\subset P_d([m],[n])$ such that $n\geq m^{1-3\epsilon/4}$
and $\mathcal{H}$ has an approximation ratio of $m^{1-\epsilon}$ w.r.t. single minded valuations. Then, there exist subsets $S\subset [m]$ and $A\subset [n]$ of sizes  $\tilde{\Omega}(m^{3\epsilon/4})$ and $m^{\epsilon/4}$ that are shattered by $\mathcal{H}$.
\end{corollary}

\begin{corollary}\label{thm:shatter_additive}
Suppose $n\ge m^{\frac{1}{k+1}-\epsilon}$ and $\ch\subset P([m],[n])$ has an approximation ratio of $m^{\frac{1}{k+1}-\epsilon}$ w.r.t. 0/1-additive valuations. Then, there exists a shattered pair $S\subset [m],\;A\in\binom{[n]}{k}$ with $|S|\ge \frac{m^{(k+1)\epsilon}}{\log_2(m)}$.
\end{corollary}

Clearly, these corollaries assert that 
if mechanism $M$ obtains a suitable approximation ratio w.r.t. to any collection of valuations $\cf$ that contains one of these two basic valuation classes, the (appropriate) shattering bound holds. 

\vspace{0.05in}\noindent{\bf Step II: Reducing from a hard problem.} As the mechanisms considered are MIR, the mechanism can optimize exactly over the shattered set of bidders and items---an intractable task for many classes of valuations. We leverage this to prove both computational complexity and communication complexity results. We next illustrate this idea in specific contexts. See our results in the appendix for other classes of valuation functions. We prove the following theorem:

\begin{theorem}\label{thm:SM,NP}
No efficient MIR mechanism has approximation ratio $m^{1-\epsilon}$ for any constant $\epsilon>0$ w.r.t. single-minded valuations, unless NP$\subseteq$ P/poly.
\end{theorem}

Theorem~\ref{thm:SM,NP} is a corollary of the following theorem for auctions with $d$-duplicates of each item:

\begin{theorem}\label{single minded duplicates}
Under the UGC, for every constant $d$ and every $\epsilon>0$, no efficient VCG-based mechanism achieves an approximation ratio of $m^{1-\epsilon}$ w.r.t. single-minded valuations unless  $NP\subseteq P/poly$.
\end{theorem}

\begin{proof}
Suppose that $M$ is MIR and obtains an approximation ratio of $m^{1-\epsilon}$  for  combinatorial auctions with $d$ duplicates w.r.t. single-minded valuations. Let $\mathcal{H}\subset P_d ([m],[n])$ denote its bank of allocations, and let $S\in P([m],[n])$ represent single minded bidders (bidder $i$ desires the bundle $S_i$). The approximation ratio implies that there must be some $T\in \mathcal{H}$ such that $n\leq m^{1-\epsilon}\cdot | \{i|S_i\subset T_i\}|~.$ Therefore, by corollary \ref{thm:shatter_single_minded}, there is a shattered pair $B\subset [n]$, $X\subset [m]$ of sizes $m^{\epsilon/4}$, $\Omega(m^{3\epsilon/4})$.

We will show computational hardness by a reduction from the \textbf{k-packing promise problem}. In this problem, the input consists of $r$ sets $S_1,....,S_r\subset U$ and
an integer $C>0$, and the goal is to distinguish between the following two options:
\begin{itemize}
	\item \textbf{Positive:} There are $C$ sets in $S_1,....,S_n$ that are pairwise disjoint.
	\item \textbf{Negative:} For every collection $\mathcal{F}$ of $C$ sets out of $S_1,....,S_n$, there is some $x\in U$ that is covered by $k$ sets of $\mathcal{F}$.
\end{itemize}
Observe that setting $k=2$ yields the famous "SET PACKING" problem, which is known to be NP-hard. We shall make use of the fact that assuming UGC, this problem is NP-hard for every constant $k$. (The proof, based on \cite{bansal2010inapproximability}, can be found in the appendix).

The reduction is given by taking $k=d+1$, and using the shattered $X$ set as the universe $U$: Given $S_1,...,S_{m^{\epsilon/4}}\subset X$ and an integer $C>0$, suppose that each bidder $i\in B$ desires $S_i$, and all bidders not in $B$ have the zero valuation.
Now, in the ``YES'' cases, since there is a collection of $C$ pairwise disjoint subsets from $S_1,...,S_{m^{\epsilon/4}}$, and since the pair $X,B$ is shattered and $M$ is MIR, the social welfare, $\sum_{i} v_i(S_i)$, will be at least $C$. On the other hand, in the ``NO'' cases, it is impossible to satisfy $C$ bidders by {\em any} $d$-duplicate allocation, and therefore, the social welfare will be less that $C$.
\end{proof}

We prove the analogue of Theorem~\ref{single minded duplicates} in the communication complexity model. We first define \textbf{multi-minded valuations} (0/1), which are are all valuations $v:2^{[m]}\to\{0,1\}$.

\begin{theorem} \label{multi-minded thm}
For every constant $d$ and $\epsilon>0$, any VCG-based mechanism for combinatorial auctions with $d$-duplicates that obtains an approximation ratio of $m^{1-\epsilon}$ w.r.t. multi-minded valuations must use exponential communication.
\end{theorem}

\begin{proof}
The proof of Theorem~\ref{multi-minded thm} involves a reduction from the communication \textbf{approximate disjointness problem} and appears in Appendix~\ref{multi-minded thm_apx}.
\end{proof}

\section{Inapproximability for General Mechanisms}\label{sec:general-techniques}

We next describe our approach to proving inapproximability results for unrestricted deterministic mechanisms. We illustrate these ideas by outlining the proof of a new separation gap between truthful and nontruthful deterministic mechanisms for combinatorial auctions with general valuations (see full proof in Appendix~\ref{sec:general_comp_apx}). Our proof applies the ``direct hardness approach'', introduced by Dobzinsky \cite{dobzinski2011impossibility} together with the VC dimension arguments.

We first define the specific valuation class that we will consider.
\begin{definition}
A valuation $v:2^{[m]}\to \mathbb R$ is called {\bf $k$-local} if there is  $T\in\binom{[m]}{k}$ such that, for all $S$, $v(S\cap T)\ge v(S)-\frac{1}{2m^2}v([m])$.
\end{definition}
For $k=k(m)$ we denote that class of $k$-local valuations by $\mathbf{k-local}$. The corresponding bidding language will be circuits.

\begin{theorem}\label{thm:general valuations}
For any $\epsilon>0$,
\begin{itemize}
\item Unless $NP\subseteq P/poly$, no efficient truthful mechanism has an approximation ratio of $m^{1-\epsilon}$ w.r.t. $m^\epsilon$-local valuations.
\item There is an efficient non-truthful mechanism with an approximation ratio of $2k$ w.r.t. $k$-local valuations. Moreover, the auctioneer communicates with the bidders only trough value queries.
\end{itemize}
\end{theorem}
We start by describing the efficient non-truthful mechanism.
\begin{algorithm}[th]\caption{} \label{alg:non_truth}
\begin{algorithmic}[1]
\STATE {\bf Input:} Oracle access to $k$-local valuations $v_1,\ldots,v_n$.
\STATE {\bf Output:} An allocation $\{S_i\}_{i\in [n]}$
\STATE Initialize $U=[m]$ and $B=[n]$.
\WHILE {$U\ne \emptyset$ and $B\ne \emptyset $}
\STATE Choose $i\in B$ that maximizes $v_{i}(U)$.
\STATE\label{step:1} Allocate $i$ a subset $S_{i}\subset U$ of size $k$ with $v_i(S_i)\ge v_i(U)-\frac{1}{2m}v_i([m])$
\STATE Set $B:=B\setminus\{i\}$ and $U:=U\setminus S_i$
\ENDWHILE
\end{algorithmic}
\end{algorithm}

\begin{lemma}
Algorithm \ref{alg:non_truth} achieves an approximation ratio of $2k$  w.r.t. $k$-local valuations.
\end{lemma}

Next, we prove the second part of theorem \ref{thm:general valuations}.
We begin by defining a class of valuations:

\begin{definition}
A valuation $v:2^{[m]}\to \mathbb R$ is called {\bf almost single minded} if there is  $T\subset [m]$ such that $\forall S\subset [m],\;\;v(S)=1[T\subset S]+\frac{1}{m^3}|S|$.
\end{definition}

These valuations will serve us in showing existence of a big \textit{menu} for one of the bidders: Given the valuations $v_{-i}$ of all bidders except bidder $i$, the \textit{menu} $R_{v_{-i}}$ of bidder $i$, is defined to be all possible bundles that may be allocated to $i$, that is
\[
R_{v_{-i}}:=\{T\subset [m]|\exists v_i\text{ for which }T\text{ is allocated to $i$ under the valuations } v_i,v_{-i}\}
\]
We recall the {\em taxation principle}: Given valuations $v_{-i}$, there is a nondecreasing function $p_{v_{-i}}:R_{v_{-i}}\to \mathbb R_+$ such that the $i$'th bidder is allocated a set that maximizes $v_i(S)-p_{v_{-i}}(S)$ over all the sets in $R_{v_{-i}}$ and pays $p_{v_{-i}}(S)$.

\begin{definition}[structured menu, \cite{dobzinski2012computational}]
Given valuations $v_{-i}$, a number $0\le k\le m$ and $0\le p$ we define the {\em structured submenu} $\cs(v_{-i},k,p)$ as all the sets $S\subset [m]$ with
\begin{itemize}
\item $S\in R_{v_{-i}}$, $|S|=k$, $p-\frac{1}{m^5}< p_{v_{-i}}(S)\le p$
\item For all $T\in R_{v_{-i}}$ such strictly contains $S$, $p_{v_{-i}}(T)\ge p_{v_{-i}}(S)+\frac{1}{m^3}$
\end{itemize}
\end{definition}

\begin{lemma}\label{lem:shatter_truthful}
Let $M$ be a mechanism with an approximation ratio of $m^{1-\epsilon}$ w.r.t. almost single-minded, $m^{\epsilon}$-local valuations and assume that $n\ge m^{1-\frac{1}{2}\epsilon}$.
There exists a bidder $i$, almost single minded $m^{\epsilon}$-local valuations $v_{-i}$, and numbers $1\le k\le m$ and $0\le p\le 2$ such that $|\cs(v_{-i},k,p)|\ge \frac{2^{\frac{m^{\frac{1}{2}\epsilon}}{2}}}{2m^6}$.
\end{lemma}

We proceed by applying the Sauer-Shelah Lemma, which provides us with a subset of items $X\subseteq[m]$ of size
$|X|=\Omega\left(m^{\frac{\epsilon}{4}}\right)$ that is shattered by some $\cs(v_{-i},k,p)$. The fact that $X$ is shattered by a structured submenu with the same level of prices, allows us to solve a hard computational problem by considering a valuation that relies on these prices. In order to define such valuation, we extend the prices from $R_{v_{-i}}$ to all subsets of $[m]$. We do so as follows:
Given  $T\subset[m]$,
set $p_{v_{-i}}(T):=\min_{T\subset S,\, S\in R_{v_{-i}}}p_{v_{-i}}(S)$.
We will need the following two facts:
\begin{lemma}[\cite{dobzinski2012computational}]\label{lem:truthful_comp}
Let $M$ be an efficient truthful mechanism and let $v_{-i}$ be some valuations with polynomial description. Then membership in $\cs(v_{-i},k,p)$ and calculation of (the extended) $p_{v_{-i}}$ can be realized by a circuit of polynomial size.
\end{lemma}
We can now prove the second part of Theorem~\ref{thm:general valuations}.

\begin{lemma}\label{vertex cover}
No efficient deterministic truthful mechanism provides an approximation ratio of $m^{1-\epsilon}$ w.r.t. $m^{\epsilon}$-local valuations unless $NP\subseteq P/poly$.
\end{lemma}

Our proof of Lemma~\ref{vertex cover} relies on a reduction from VERTEX-COVER: Given a graph whose vertices are the items in the shattered sets, we define the valuation $v_i$ such that $v_i$ assigns high values to all vertex covers but, through the menu prices, makes covers of size $\geq t$ unattractive. We then leverage the taxation principle to show that if there is a cover $C$ of size $t$ of $U$, the mechanism will allocate bidder $i$ a set $T$ such that $T\cap U$ is a cover of size $t$.

\subsubsection*{Acknowledgements}
Amit Daniely's research was supported, in part, by the Google Europe Fellowship in Learning Theory. We thank Subhash Khot for referring us to~\cite{bansal2010inapproximability} and pointing out that the problem studied there can be reduced to $k$-packing promise problem (see the proof of Theorem~\ref{single minded duplicates}).
We thank Guy Kindler and Nati Linial for valuable discussions.

\bibliography{bib}

\section*{Appendix}

\appendix

\section{Proof of Theorem \ref{thm:sauer_gen}}\label{apx_sauer_gen}

Our proof of Theorem \ref{thm:sauer_gen} (the generalized Sauer-Shelah Lemma) relies on a careful application of techniques used for proving previous bounds (e.g., \cite{steele1978existence,Natarajan89b,haussler1995generalization}). We explain the interesting connections between this VC machinery and so called ``approximate-partitions'' in Section~\ref{sec:partitions}. We use the following notation:

\begin{definition}
For $m\ge d\ge 0$ and $n\ge k\ge 2$ let $M_{n,k}(d,m)$ be the maximal cardinality of a class $\ch\subset [n]^{[m]}$ with $\Dim_k(\ch)\le d$.
\end{definition}
Theorem \ref{thm:sauer_gen} can be restated as
\[
M_{n,k}(d,m)\le \sum_{i=0}^{d} \binom{m}{i} (k-1)^{m-i}\binom{n}{k}^{i}~.
\]
We use the following recursive inequality:
\begin{lemma}\label{lem:recursive}
For every $2\le k\le n$ and $m>d>0$ we have
\[
M_{n,k}(d,m)\le (k-1)\cdot M_{n,k}(d,m-1)+\binom{n}{k}\cdot M_{n,k}(d-1,m-1)
\]
\end{lemma}
\begin{proof}
Let $\ch\subset [n]^{[m]}$ be a class of $k$-dimension $d$ with $|\ch|=M_{n,k}(d,m)$. For every subset $Y\subset [n]$ of size $k$ let:
\begin{itemize}
\item $\ch_{Y,1}$ be the class of all function $f\in \ch$ such that $f(m)\in Y$ and for every $y\in Y$ there is $f'\in \ch$ such that $f|_{[m-1]}=f'|_{[m-1]}$ and $f'(m)=y$.
\item $\ch_{Y,2}$ be the class of all function $f\in \ch_{Y,1}$ such that $f(m)= \max(Y)$
\end{itemize}
Finally, let
\begin{itemize}
\item $\ch'=\ch\setminus\left(\cup_{Y\in\binom{[n]}{k}}\ch_{Y,2}\right)$.
\end{itemize}
Clearly,
\[
|\ch|\le |\ch'|+\sum_{Y\in\binom{[n]}{k}}|\ch_{Y,2}|~.
\]
The proof of the lemma hence follows from the following two claims:

{\bf Claim 1: $|\ch'|\le M_{n,k}(d,m-1) (k-1)$.} Indeed, for every function $f\in \ch'$ there are at most $k-1$ functions $f'\in \ch'$ with $f'|_{[m-1]}=f|_{[m-1]}$. Otherwise, there are $k$ different $f_1,\ldots,f_k\in\ch'$ that coincide on $[m-1]$. This implies a contradiction as one of these functions belongs to $\ch_{\{f_1(m),\ldots,f_k(m)\},2}$. Therefore, we conclude that $|\ch'|\le |\ch'_{[m-1]}|\cdot (k-1)\le M_{n,k}(d,m-1)\cdot (k-1)$.

{\bf Claim 2: for every $Y\in\binom{[n]}{k}$, $|\ch_{Y,2}|\le M_{n,k}(d-1,m-1)$.} Indeed, we note that $|\ch_{Y,2}|=|\ch_{Y,2}|_{[m-1]}|=|\ch_{Y,1}|_{[m-1]}|$. We also note that $\Dim_{k}(\ch_{Y_1}|_{[m-1]})=d-1$. Otherwise, we could add $m$ to a $d$-elements subset of $[m-1]$ that is $k$-shattered to obtain a $(d+1)$-elements subset of $[m]$ that is $k$-shattered.
\end{proof}
The next simple lemma calculates the values of $M_{n,k}(d,m)$ for the pairs $d,m$ for which the above recursion inequality does not apply.
\begin{lemma}\label{lem:initial}
For every $2\le k\le n$ and every $m\ge 1$ we have
\begin{itemize}
\item $M_{n,k}(0,m)=(k-1)^{m}$.
\item $M_{n,k}(m,m)=n^{m}$.
\end{itemize}
\end{lemma}
The proof of the lemma is straightforward.
We will upper bound $M_{n,k}(d,m)$ in terms of the numbers $N_{n,k}(d,m)$ that are defined recursively as follows
\[
N_{n,k}(d,m)=(k-1)\cdot N_{n,k}(d,m-1)+\binom{n}{k}\cdot N_{n,k}(d-1,m-1)
\]
With the initial conditions $N_{n,k}(m,m)=n^m$ and $N_{n,k}(0,m)=(k-1)^{m}$. By lemmas \ref{lem:recursive} and \ref{lem:recursive} it is clear that $M_{n,k}(m,d)\le N_{n,k}(m,d)$. Therefore, theorem \ref{thm:sauer_gen} follows from the following lemma.
\begin{lemma}
For every $m\ge d\ge 1$ and $n\ge k$ we have
\begin{equation}\label{eq:2}
N_{n,k}(d,m)\le \sum_{i=0}^{d} \binom{m}{i} (k-1)^{m-i}\binom{n}{k}^{i} ~.
\end{equation}
\end{lemma}
\begin{proof}
We use induction on $m+d$. We first deal with the cases $d=0$ and $d=m$. The case $d=0$ is clear. For the case $d=m$, we split into two sub-cases: $n=k$ and $k<n$. If $k<n$ then the last term, in the r.h.s. of equation (\ref{eq:2}) upper bounds $N_{d,m}$. If $k=n$ then the r.h.s. of equation (\ref{eq:2}) counts the number of functions in $[n]^{[d]}$, and therefore equal to $N_{n,k}(d,m)=n^d$. To see that, suppose that in order to choose a function in $[n]^{[d]}$, we first choose the pre-image of $1$ and after that, we assign values in $\{2,\ldots,n\}$ to the remaining points in $[d]$.

We proceed to the induction step. After dealing with the cases $d=0$ and $d=m$, we may assume that $m>d\ge 1$. Now, by the induction hypothesis and the definition of $N_{n,k}(d,m)$ we have
\begin{eqnarray*}
N_{n,k}(d,m) &=& (k-1)\cdot N_{n,k}(d,m-1)+\binom{n}{k}\cdot N_{n,k}(d-1,m-1)
\\
&\le &
(k-1)\cdot \sum_{i=0}^{d} \binom{m-1}{i} (k-1)^{m-i-1}\binom{n}{k}^{i}
+
\binom{n}{k}\cdot \sum_{i=0}^{d-1} \binom{m-1}{i} (k-1)^{m-i-1}\binom{n}{k}^{i}
\\
&= &
\sum_{i=0}^{d} \binom{m-1}{i} (k-1)^{m-i}\binom{n}{k}^{i}
+
\sum_{i=0}^{d-1} \binom{m-1}{i} (k-1)^{m-i-1}\binom{n}{k}^{i+1}
\\
&= &
\binom{m-1}{0} (k-1)^{m}\binom{n}{k}^{0}
+
\sum_{i=1}^{d} \left[ \binom{m-1}{i} +\binom{m-1}{i-1}\right](k-1)^{m-i}\binom{n}{k}^{i}
\\
&= &
\binom{m}{0} (k-1)^{m}\binom{n}{k}^{0}
+
\sum_{i=1}^{d}  \binom{m}{i} (k-1)^{m-i}\binom{n}{k}^{i}
\\
&= &
\sum_{i=0}^{d}  \binom{m}{i} (k-1)^{m-i}\binom{n}{k}^{i}
\end{eqnarray*}
\end{proof}

\section{Background: Mechanisms and Combinatorial Auctions}\label{sec:background}

\subsection{Mechanisms for Auctions}

In a \textit{combinatorial auction}, $m$ items are offered for sale to $n$ bidders. Each bidder $i\in [n]$ has a private valuation function $v_i:2^{[m]}\rightarrow[0,
\infty)$ that is \textit{normalized} (i.e. $v_i(\emptyset)=0$) and \textit{nondecreasing} (if $S\subset T$ then $v_i(S)\leq v_i(T)$). A {\em mechanism} is an algorithm that takes as input bidders' valuation functions (either directly or via oracle access) and outputs an {\em allocation} (i.e., a collection $S=\{S_i\}_{i\in [n]}$ of pairwise disjoint subsets of $[m]$) of items to the bidders, and a vector of prices $p_i\ge 0\;\; i\in [n]$ that each bidder has to pay.

We will consider combinatorial auctions under different assumptions on the valuation functions, and also {\em combinatorial auction with duplicates}: $d$ identical units of each item are available but each bidder desires at most 1 unit of each item, and the mechanism can output a {\em $d$-duplicate} allocation, i.e., a collection $S=\{S_i\}_{i\in [n]}$ of subsets of $[m]$ such that each $x\in [m]$ belongs to at most $d$ of these subsets.

Important desiderata in this context are:

\begin{itemize}
\item \textbf{Good outcomes}. We seek mechanisms that provide ``good'' approximations to the optimum social welfare, i.e., to the allocation $S^*=\{S^*_i\}_{i\in [n]}$ that maximizes the expression $\Sigma_i v_i(S^*_i)$.
\item \textbf{Truthfulness (a.k.a. incentive compatibility)}. Under truthful mechanisms, a bidder maximizes his utility $v_i(S_i)-p_i$ by reporting his actual valuation to the mechanism regardless of other bidders' reports (that is, truthfulness is a dominant strategy for each bidder).
\item \textbf{Tractability}. We will consider both the standard computational complexity model and the query complexity model (see Introduction).
\end{itemize}

\begin{definition}[Maximal-in-range (MIR)/VCG based mechanism]
A mechanism $M$ is called {\em MIR} if there exists a {\em bank of $d$-duplicate allocations} $\ch$ such that the output of $M$ for every $n$-tuple of valuation function $v_1,\ldots,v_n$ is an allocation $S$ that maximizes $\sum_{i=1}^n v_i(S_i)$ over all allocations $S\in\ch$.
\end{definition}

\subsection{Valuation Classes and Complexity Models}

Research on combinatorial auctions investigates the implications of different restrictions on valuation functions. We consider the following classes of valuations:
\begin{itemize}
\item $\mathbf{General}$: all (non-negative, nondecreasing and normalized) valuations.
\item $\mathbf{Multi-Minded}$ (0/1): all valuations $v:2^{[m]}\to\{0,1\}$.
\item $\mathbf{Single-Minded}$ (0/1): all valuations of the form $v(S)=1[T\subset S]$ for some $T\subset [m]$.
\item $\mathbf{Subadditive}$: all valuations satisfying $v(S\dot\cup T)\le v(S)+v(T)$.
\item $\mathbf{XOS}$: all valuations of the form $v(S)=\max_{j}w_j^T\chi(S)$, where each $w_j$ is a vector in $\mathbb{R}_+^n$ (here, $\chi(S)$ is the characteristic vector of $S$).
\item $\mathbf{Submodular}$: all valuations satisfying $v(S\cup T)\le v(S)+v(T)-v(S\cap T)$.
\item $\mathbf{Additive}$. all valuations of the form $v(S)=w^T\chi(S)$, where $w\in\mathbb{R}_+^n$.
\item $\mathbf{0/1-Additive}$. all valuations of the form $v(S)=w^T\chi(S)$, where $w\in\{0,1\}^n$.
\end{itemize}

These classes form the following hierarchies~\cite{nisan2007algorithmic}:
\[
\mathbf{0/1-Additive}\subset \mathbf{Additive}\subset\mathbf{Submodular}\subset \mathbf{XOS} \subset \mathbf{Subadditive}\subset\mathbf{General}
\]
and
\[
\mathbf{Single-Minded}\subset \mathbf{Multi-Minded}\subset\mathbf{General}~.
\]

All inclusions above are strict.

A naive representation of the ``input'' to combinatorial auctions---the valuation functions---requires specifying $2^m$ numbers. In contrast, mechanisms are expected to be efficient in the natural parameters of the problem: the number of bidders $n$ and the number of items $m$. We consider the two standard models for accessing the in combinatorial auctions:

\begin{itemize}

\item {\bf The computational complexity model}, in which valuation functions are succinctly encoded and mechanisms must run in time that is polynomial in the input length; for single-minded bidders, for instance, the encoding of valuation function $v_i$ is simply the bundle of items bidder $i$ desires.
    
\item {\bf The oracle model}, in which valuations are treated as black boxes that can only answer a certain type of queries, and complexity is measured in terms of the number of queries. Three types of queries are commonly considered: (i) a \emph{value query} to valuation $v_i$ is a subset $S\subseteq [m]$ and the reply is the value $v_i(S)$; (ii) a \emph{demand query} to $v_i$ is a vector $p = (p_1, \ldots p_m)$ of ``item prices'', and the reply is bidder $i$'s ``demand'' at these prices, that is, a subset $S$ that maximizes the expression $v_i(S) − \Sigma_{j\in S} p_j$; and (iii) the \emph{communication complexity} model, in which provide a unifying framework for essentially all existing inapproximability results for combinatorial auctions to dateoracles can answer any kind of query (addressed to a single valuation) and the complexity of the mechanism is measured in terms of the number of bits transmitted.

\end{itemize}

\section{Approximation w.r.t. Valuations Classes}\label{sec:results-first}

\begin{definition}
Fix a collection of valuations $\cf$. Let $\ch\subset P_d(X,Y)$ be a collection of allocations. We say that $\ch$ has {\bf approximation ratio of $\alpha\ge 1$ w.r.t. to $\cf$} if for every choice of valuations $\{v_y\}_y\in Y\subset \cf$ and every partition $\{S_y\}_{y\in Y}\in P(X,Y)$ there is an allocation $\{T_y\}_{y\in Y}\in \ch$ such that
\[
\sum_{y\in Y}v_y(T_y)\ge \frac{1}{\alpha}\cdot \sum_{y\in Y}v_y(S_y)
\]
\end{definition}
It is easy to see that $\ch$ has an approximation ratio of $\alpha$ w.r.t. $\cf$ if and only if running MIR mechanism using $\ch$ as the bank of allocations have an approximation ratio of $\alpha$ w.r.t. $\cf$.
We will consider questions of the following type:

\begin{question}\label{que:1}
Given a collection $\cf$ of valuation, what is the best (i.e. smallest) approximation ratio $\alpha\ge 1$ for which there exists a collection of allocations with approximation ratio of $\alpha$ and no large shattered pair of sets?
\end{question}

\subsection{Single-Minded Valuations}
Fix a set of items, $X$, of size $m$ and a set of bidders, $Y$, of size $n$.
In this section we consider question \ref{que:1} when $\cf$ is the class of general (i.e., nondecreasing and normalized) valuations or $\cf$ is the class of single minded valuations.

Consider first that case that $\cf$ is the class of general valuations and consider the allocation class $\ch$ consisting of all allocations that give {\em all} the items to a single player. That is, $\ch$ consists of all constant functions $f:X\to Y$. It is not hard to see $\ch$ does not shatter any pair $A\subset Y,\;S\subset X$ where $|S|\ge 2$. Also, it is not hard to see that $\ch$ has an approximation ratio of $m$ (even $\min(m,n)$) w.r.t. general valuations.

The following theorem shows that if we require bit more, that is, an approximation ratio of $m^{1-\epsilon}$ w.r.t. {\em just single minded valuations}, then we can already find a shattered pair $A,S$, where both $A$ and $S$ are large.

\begin{theorem}\label{thm_apx:shatter_single_minded}
Let  $\epsilon>0$, and let $\mathcal{H}\subset P_d([m],[n])$ such that $n\geq m^{1-3\epsilon/4}$
and $\mathcal{H}$ has an approximation ratio of $m^{1-\epsilon}$ w.r.t. single minded valuations. Then, there exist subsets $S\subset [m]$ and $A\subset [n]$ of sizes  $\tilde{\Omega}(m^{3\epsilon/4})$ and $m^{\epsilon/4}$ that are shattered by $\mathcal{H}$.
\end{theorem}
The theorem follows immediately from theorem \ref{thm:shatter_single_minded_0}, as we note that $\ch$ has the $\alpha$ containment property if and only if $\mathcal{H}$ has an approximation ratio of $m^{1-\epsilon}$ w.r.t. single minded valuations.

\subsection{0/1-Additive, Submodular, XOS and Subadditive Valuations}
Fix a set of items, $X$, of size $m$ and a set of bidders, $Y$, of size $n$.
In this section we consider question \ref{que:1} for additive, submodular, XOS and subadditive valuations. Recall that
\[
\mathbf{0/1-Additive}\subset \mathbf{Submodular}\subset \mathbf{XOS} \subset \mathbf{Subadditive}
\]
and all the inclusions are strict. In \cite{dobzinski2005approximation} it is shown that there is a class\footnote{This class is simply the collection of all allocations that either gives all items to a single bidder, or, gives at most a single item to each bidder.} $\ch\subset P(X,Y)$ with an approximation ratio of $m^{\frac{1}{2}}$ (even $\min\left(m^{\frac{1}{2}},n\right)$) w.r.t. subadditive valuations and with no shattered pair $S\subset X,a\in Y$ with $|S|\ge 3$.
Using this class \cite{dobzinski2005approximation} provided truthful mechanism with approximation ratio of $m^{\frac{1}{2}}$ w.r.t. subadditive valuations. As shown in \cite{buchfuhrer2009limits}, if we require bit more, namely, an approximation ratio of $m^{\frac{1}{2}-\epsilon}$ w.r.t. just to 0/1-additive valuations, we can find a shattered pair $S\subset X,a\in Y$ for a large set $S$.

As we show next, if we require a stronger approximation ratio, of $m^{\frac{1}{3}-\epsilon}$, we can find a shattered pair $S\subset X,A\subset Y$ with $|A|=2$ and a large set $S$. This will be useful later to establish communication lower bounds.
\begin{theorem}\label{thm:shatter_additive_apx}
Suppose $n\ge m^{\frac{1}{k+1}-\epsilon}$ and $\ch\subset P(X,Y)$ has an approximation ratio of $m^{\frac{1}{k+1}-\epsilon}$ w.r.t. 0/1-additive valuations. There is a shattered pair $S\subset X,\;A\in\binom{Y}{k}$ with $|S|\ge \frac{m^{(k+1)\epsilon}}{\log_2(m)}$.
\end{theorem}
The theorem follows immediately from theorem \ref{thm:shatter_additive_0}, as we note that $\ch$ has the $\alpha$ intersection property if and only if $\mathcal{H}$ has an approximation ratio of $m^{1-\epsilon}$ w.r.t. single 0/1-additive valuations.
Theorem \ref{thm:shatter_additive_0} shows that if a class have a good approximation ratio w.r.t. 0/1-additive, we can find a set of constantly many bidders $A$ and a large subset $S\subset X$ such that the pair $S,A$ is shattered.
The following theorem shows that for arbitrarily good approximation ratio $\alpha>1$, there are classes with an approximation ratio of $\alpha$ w.r.t. XOS valuations, with no shattered pair $S,A$ where $S$ is large (say, super-logarithmic) and $|A|=\omega(\log(n))$. This indicates on limitation of the VC technique for establishing computational lower bound on combinatorial auctions with XOS valuations.

\begin{theorem}
For every $l\le |Y|$ there is a class $\ch\subset Y^X$ such that
\begin{itemize}
\item $\ch$ has an approximation ratio of $\frac{l}{l-1}$ w.r.t. XOS valuations.
\item For $l=2$, $\ch$ has an approximation ratio of $2$ w.r.t. subadditive valuations.
\item There is no shattered pair $S\subset X$ and $A\subset Y$ with $|S|\ge 8(l-1)\ln(n)$ and $|A|>24(l-1)\ln(m)$.
\end{itemize}
\end{theorem}
\begin{proof}
For simplicity, we assume that $l$ divides $|Y|$.
Fix $l\ge 2$ and let $\cy=\{Y_{x,j}\}_{x\in X,1\le j\le l}$ be non-empty subsets of $Y$ such that for every $x\in X$, $Y=Y_{x,1}\dot\cup\ldots\dot\cup Y_{x,l}$. Define $\ch^{\cy}$ to be the class of all functions $f:X\to Y$ for which there is $1\le i\le l$ such that $\forall x,\;f(x)\not \in Y_{i}$. As the following claim shows, for every $\cy$, $\ch^{\cy}$ satisfies the first two conclusions of the theorem. Later on, we will show that for some $\cy$, $\ch^{\cy}$ satisfies also the last conclusion.
\begin{claim} The following conditions hold:
\begin{itemize}
\item $\ch^{\cy}$ has an approximation ratio of $2$ w.r.t.
subadditive valuations.
\item $\ch^{\cy}$ has an approximation ratio of $\frac{l}{l-1}$ w.r.t. XOS valuations.
\end{itemize}
\end{claim}
\begin{proof}
We prove the second part. The proof of the first part is similar. Let $\{v_y\}_{y\in Y}$ be a collection of XOS valuation and let $\{S_y\}_{y\in Y}\in P(X,Y)$. We must show that there is $\{T_y\}_{y\in Y}\in \ch^{\cy}$ with
\[
\sum_{y\in Y}v_y(T_y)\ge \frac{l-1}{l} \sum_{y\in Y}v_y(S_y)~.
\]
For every $1\le i\le l$ let $\{S^i_y\}_{y\in Y}$ be the allocation defined by
\[
S^i_y=S_y\setminus\{x\in X\mid y\in Y_{x,i}\}~.
\]
The allocation $\{S^i_y\}_{y\in Y}$ is contained in one of the allocations in $\ch^{\cy}$. To see that, let $f:X\to Y$ be any function for which $f|_{S^i_y}$ is the constant function $y$ and $f(x)\not\in Y_{x,i}$ for every $x\in \cup_{y\in Y}\{x'\in S_y\mid y\in Y_{x',i}\}$. It is not hard to check that $f\in \ch^{\cy}$ and for all $y\in Y$, $S^i_{y}\subset f^{-1}(y)$.

Therefore it is enough to show that for some $i$,
\[
\sum_{y\in Y}v_y(S^i_y)\ge \frac{l-1}{l} \sum_{y\in Y}v_y(S_y)~.
\]
Indeed, we will show that for random $i$, the expected value of the l.h.s. is at least the r.h.s. To this end, it is enough to show that for any fixed $y\in Y$,
\[
\frac{1}{l}\sum_{i=1}^lv_y(S^i_y)\ge \frac{l-1}{l} v_y(S_y)~.
\]
suppose that $v_y(U)=\max_{j}w_j^T\chi(U)$. Let $j^*$ be the index for which $v_y(S_y)=w_{j^*}^T\chi(S_y)$. Now, we have
\begin{eqnarray*}
\sum_{i=1}^lv_y(S^i_y) &\ge& \sum_{i=1}^l w_{j^*}^T\chi(S_y^i)=(l-1)\cdot w_{j^*}^T\chi(S_y)=(l-1)v_y(S_y)
\end{eqnarray*}
where the left equality follows from the fact that each element $x\in S_y$ appears in exactly $l-1$ many $S^i_y$'s (namely, all $i$'s except the one for which $y\in Y_{x,i}$).
\end{proof}
Suppose now that for every $x\in X$, the partition $\{Y_{x,i}\}_{i\in [l]}$ is chosen at random, uniformly from all partitions of $Y$ into $l$ non empty subsets, each of which is of size $\frac{|Y|}{l}$. Now, fix subsets $S\subset X$ and $A\subset Y$.
\begin{claim}
The probability that the pair $S,A$ is shattered is $\le |S|^ll^{|S|}\left(\frac{l-1}{l}\right)^{|A|(|S|-l)}$.
\end{claim}
\begin{proof}
We say that $x\in S$ is a {\em bad point} if for every $1\le i\le l$, $A\cap Y_{x,i}\ne \emptyset$. We note that if there are $l$ bad points, then the pair $S,A$ is not shattered. Indeed, if $x_1,\ldots,x_l$ are bad points, then we can choose for every $1\le i\le l$ some $y_{x_i}\in A\cap Y_{x,i}$. It is not hard to see that there is no $f\in\ch^{\cy}$ satisfying $\forall i,\;f(x_i)=y_{x_i}$. Therefore, the pair $S,A$ is not shattered.

It is therefore enough to bound the probability that there are less than $l$ bad points. Indeed, for every $x\in S$, the probability that $A\cap Y_{x,i}=\emptyset$ is $\le\left(\frac{l-1}{l}\right)^{|A|}$. Hence, the probability that $x$ is good is $\le l\left(\frac{l-1}{l}\right)^{|A|}$. It follows that the probability that there are less than $l$ bad points is
\[
\le \sum_{i=1}^{l-1}\binom{|S|}{i}l^{|S|-i}\left(\frac{l-1}{l}\right)^{|A|(|S|-i)}\le |S|^ll^{|S|}\left(\frac{l-1}{l}\right)^{|A|(|S|-l)}
\]
\end{proof}
By the claim, it follows that the probability that some pair $S,A$ with $|S|=m'>8(l-1)\ln(n)$ and $|A|=n'>24(l-1)\ln(m)$ is shattered is
\[
\le \binom{n}{n'}\binom{m}{m'}m'^ll^{m'}\left(\frac{l-1}{l}\right)^{n'(m'-l)}\le n^{n'}m^{2m'+l}\left(\frac{l-1}{l}\right)^{n'(m'-l)}~.
\]
Taking $\ln$, and using the facts that $2l\le m'$ and that for every $0\le x\le 1$ we have that $\ln(1+x)\ge \ln(2)x\ge \frac{x}{2}$, we conclude that the $\ln$ of the last probability is
\begin{eqnarray*}
&&n'\ln(n)+(2m'+l)\ln(m)-\ln\left(\frac{l}{l-1}\right)n'(m'-l)
\\
&\le&
n'\ln(n)+(2m'+l)\ln(m)-\frac{1}{2(l-1)}n'(m'-l)
\\
&\le&
n'\ln(n)+3m'\ln(m)-\frac{1}{4(l-1)}n'm'
\\
&=&
n'\left(\ln(n)-\frac{m'}{8(l-1)}\right)+m'\left(3\ln(m)-\frac{n'}{8(l-1)}\right)<0~.
\end{eqnarray*}
Therefore, there is a positive probability that no such pair is shattered. Hence, there is some collection of partitions, $\cy$, for which $\ch^{\cy}$ satisfies the third conclusion of the theorem.
\end{proof}

\section{General Valuations}\label{sec:general_comp_apx}
In this section, we will show that unless $NP\subset P/poly$, there is no efficient truthful mechanism for general valuations with an approximation ratio of $m^{1-\epsilon}$. We will even show that for every $\epsilon>0$, there is a class of valuations for which efficient non-truthful mechanism can achieve an approximation ratio of $m^\epsilon$, but the approximation ratio of every efficient truthful mechanism is $\ge m^{1-\epsilon}$.
We will use the "direct hardness approach", introduced by Dobzinsky \cite{dobzinski2011impossibility},
together with the VC dimension technique.

We first define the valuation class that we will consider.
\begin{definition}
A valuation $v:2^{[m]}\to \mathbb R$ is called {\bf $k$-local} if there is  $T\in\binom{[m]}{k}$ such that, for all $S$,
\[
v(S\cap T)\ge v(S)-\frac{1}{2m^2}v([m])~.
\]
\end{definition}
For $k=k(m)$ we denote that class of $k$-local valuations by $\mathbf{k-local}$. The corresponding bidding language will be circuits.

\begin{theorem}\label{thm_apx:general valuations}
For any $\epsilon>0$,
\begin{itemize}
\item Unless $NP\subseteq P/poly$, no efficient truthful mechanism has an approximation ratio of $m^{1-\epsilon}$ w.r.t. $m^\epsilon$-local valuations.
\item There is an efficient non-truthful mechanism with an approximation ratio of $2k$ w.r.t. $k$-local valuations. Moreover, the auctioneer communicates with the bidders only trough value queries.
\end{itemize}
\end{theorem}
We start by describing the efficient non-truthful mechanism.
\begin{algorithm}[th]\caption{} \label{alg:non_truth_apx}
\begin{algorithmic}[1]
\STATE {\bf Input:} Oracle access to $k$-local valuations $v_1,\ldots,v_n$.
\STATE {\bf Output:} An allocation $\{S_i\}_{i\in [n]}$
\STATE Initialize $U=[m]$ and $B=[n]$.
\WHILE {$U\ne \emptyset$ and $B\ne \emptyset $}
\STATE Choose $i\in B$ that maximizes $v_{i}(U)$.
\STATE\label{step:1_apx} Allocate $i$ a subset $S_{i}\subset U$ of size $k$ with $v_i(S_i)\ge v_i(U)-\frac{1}{2m}v_i([m])$
\STATE Set $B:=B\setminus\{i\}$ and $U:=U\setminus S_i$
\ENDWHILE
\end{algorithmic}
\end{algorithm}

\begin{lemma}
Algorithm \ref{alg:non_truth_apx} achieves an approximation ratio of $2k$  w.r.t. $k$-local valuations.
\end{lemma}
\begin{proof}
We first explain how to implement step \ref{step:1_apx} in the algorithm using only value queries.
Given a set $U$ consisting at least
$k$ elements, we need to find a set $S\in
\binom{U}{k}$ with $v(S)\ge v(U) -\frac{1}{2m}v([m])$. This can be done as follows. Simply,
as long as $|U|\ge k$, we throw an element $x\in U$ that
maximizes $v(U\setminus\{x\})$. Since at each step we loose at most $\frac{1}{2m^2}v([m])$, the total loss is $\le \frac{1}{2m}v([m])$.

Next, we show that the algorithm indeed achieves an approximation ratio of $2k$. Let $S_1,\ldots,S_n$ be an optimal allocation. Assume, w.l.o.g. that $v_1(S_1)\ge \ldots\ge v_n(S_n)$ and denote $OPT=v_1(S_1)+\ldots+ v_n(S_n)$. Note that in the beginning of the $j$'th iteration $[m]\setminus U$ consists of $(j-1)k$ elements. Therefore, one of the sets $S_1,\ldots, S_{(j-1)k+1}$ is contained in $U$. Therefore the value of the allocation at this iteration is
\[
\ge v_{(j-1)k+1}(S_{(j-1)k+1})-\frac{1}{2m}OPT\ge
\frac{1}{k}\sum_{i=1}^kv_{(j-1)k+i}(S_{(j-1)k+1})-\frac{1}{2m}OPT
\]
(here, we used the fact that $v_1(S_1)\ge \ldots\ge v_n(S_n)$ and that $v_i([m])\le OPT$ for all bidders $i$). Summing over all iteration, and using the fact that there at most $\frac{m}{k}$ iterations, we conclude that the welfare of the allocation returned by the algorithm is $\ge \frac{1}{k}OPT-\frac{m}{k}\frac{1}{2m}OPT=\frac{1}{2k}OPT$
\end{proof}

Next, we prove the second part of theorem \ref{thm_apx:general valuations}.
We begin by defining a class of valuations:

\begin{definition}
A valuation $v:2^{[m]}\to \mathbb R$ is called {\bf almost single minded} is there is  $T\subset [m]$ such that $\forall S\subset [m],\;\;v(S)=1[T\subset S]+\frac{1}{m^3}|S|$.
\end{definition}

These valuations will serve us in order to show existence of a big \textit{menu} for one of the bidders: Given the valuations $v_{-i}$ of all bidders except bidder $i$, the \textit{menu} $R_{v_{-i}}$ of bidder $i$, is defined to be all possible bundles that may be allocated to $i$, that is
\[
R_{v_{-i}}:=\{T\subset [m]|\exists v_i\text{ for which }T\text{ is allocated to $i$ under the valuations } v_i,v_{-i}\}
\]
We recall here the {\em taxation principle}. Given valuations $v_{-i}$, there is a nondecreasing function $p_{v_{-i}}:R_{v_{-i}}\to \mathbb R_+$ such that the $i$'th bidder is allocated a set that maximizes $v_i(S)-p_{v_{-i}}(S)$ over all the sets in $R_{v_{-i}}$ and pays $p_{v_{-i}}(S)$.

\begin{definition}[structured menu, \cite{dobzinski2012computational}]
Given valuations $v_{-i}$, a number $0\le k\le m$ and $0\le p$ we define the {\em structured submenu} $\cs(v_{-i},k,p)$ as all the sets $S\subset [m]$ with
\begin{itemize}
\item $S\in R_{v_{-i}}$
\item $|S|=k$
\item $p-\frac{1}{m^5}< p_{v_{-i}}(S)\le p$
\item For all $T\in R_{v_{-i}}$ such strictly contains $S$, $p_{v_{-i}}(T)\ge p_{v_{-i}}(S)+\frac{1}{m^3}$
\end{itemize}
\end{definition}

\begin{lemma}\label{lem:shatter_truthful_apx}
Let $M$ be a mechanism with an approximation ratio of $m^{1-\epsilon}$ w.r.t. almost single minded, $m^{\epsilon}$-local valuations and assume that $n\ge m^{1-\frac{1}{2}\epsilon}$.
There exists a bidder $i$, almost single minded $m^{\epsilon}$-local valuations $v_{-i}$, and numbers $1\le k\le m$ and $0\le p\le 2$ such that $|\cs(v_{-i},k,p)|\ge \frac{2^{\frac{m^{\frac{1}{2}\epsilon}}{2}}}{2m^6}$.
\end{lemma}

\begin{proof}
For simplicity, we assume that $n=m^{1-\frac{1}{2}\epsilon}$. Given almost single minded  valuations $v_{-i}$ we denote by $\cs_{v_{-i}}$ all the sets $T\subset [m]$ that the bidder $i$ can obtain by biding with an almost single minded valuation.
\begin{claim}[\cite{dobzinski2011impossibility}]
\begin{itemize}
\item For every $S\in\cs_{v_{-i}}$ and every $T\in R_{v_{-i}}$ that strictly contains $S$, $p_{v_{-i}}(T)\ge p_{v_{-i}}(S)+\frac{1}{m^3}$.
\item For every $S\in\cs_{v_{-i}}$, $0\le p_{v_{-i}}(S)\le 2$.
\end{itemize}
\end{claim}
By this claim, it is enough to show that for some $v_{-i}$, $|\cs_{v_{-i}}|\ge \frac{2^{\frac{m}{2n}}}{m}$. Indeed, in that case, by the claim, we have
\[
S_{-v_i}\subset\cup_{0\le k\le m,0\le p'\le 2m^5 }\cs\left(v_{-i},k,\frac{p'}{2m^5}\right)~.
\]
By the pigeonhole principle, there is some $\cs\left(v_{-i},k,\frac{p'}{2m^5}\right)$ consisting of $\ge \frac{2^{\frac{m}{2n}}}{2m^6}$ sets.

Suppose toward a contradiction that it is not the case.
Consider random almost single minded valuations $v_1,\ldots, v_n$ sampled uniformly from all such collections for which each $v_i$ corresponds to a set $S_i$ of size $\frac{m}{2n}$ and the sets $S_i$ are mutually disjoint.

Let $\{T_i\}_{i\in[n]}$ be the allocation given by the mechanism operating on these random valuations.
We say that the bidder $i$ is satisfied if $S_i\subset T_i$.

\begin{claim}
With a positive probability, non of the bidders is satisfied by a set of size $\le \frac{m}{4}$.
\end{claim}
This entails contradiction since in that case there are at most $4$ satisfied bidders, and therefore the total welfare is $\le 5$. On the other hand, there is some allocation that satisfies all bidders, and therefore its welfare is $\ge m^{1-\frac{1}{2}\epsilon}$. This contradicts the assumption that the approximation ratio is $m^{1-\epsilon}$. It is therefore left to prove the claim:

\begin{proof}
It is enough to show that the probability that the $n$'th bidder is satisfied by a set of size $\le \frac{m}{4}$ is $<\frac{1}{n}$. Indeed, given the sets $S_1,\ldots,S_{n-1}$, the set $S_n$ is a random subset of $X'=X\setminus \left(\cup_{i=1}^{n-1}S_i\right)$. Since $|X'|\ge \frac{1}{2}m$, the probability that the $n$'th bidder is satisfied by a {\em fixed} set $U\subset X$ of size $\le \frac{n}{4}$ is at most $\left(\frac{1}{2}\right)^{\frac{m}{2n}}$. Now, since $|S_{v_{-n}}|< \frac{2^{\frac{m}{2n}}}{m}$, there are less than $ \frac{2^{\frac{m}{2n}}}{m}$ potential such sets, and the claim follows.
\end{proof}
\end{proof}

We proceed by applying the Sauer-Shelah lemma which provides us with a subset of items $X\subseteq[m]$ of size
$|X|=\Omega\left(m^{\frac{\epsilon}{4}}\right)$ that is shattered by some $\cs(v_{-i},k,p)$. The fact that $X$ is shattered by a structured submenu with the same level of prices, allows us to solve a hard computational problem by considering a valuation that relies on these prices. In order to define such valuation, we extend the prices from $R_{v_{-i}}$ to all subsets of $[m]$. We do so as follows:
Given  $T\subset[m]$,
set $p_{v_{-i}}(T):=\min_{T\subset S,\, S\in R_{v_{-i}}}p_{v_{-i}}(S)$.
We will need the following two facts:
\begin{lemma}[\cite{dobzinski2012computational}]\label{lem:truthful_comp_apx}
Let $M$ be an efficient truthful mechanism and let $v_{-i}$ be some valuations with polynomial description. Then membership in $\cs(v_{-i},k,p)$ and calculation of (the extended) $p_{v_{-i}}$ can be realized by a circuit of polynomial size.
\end{lemma}
We are now ready to prove the second part of theorem \ref{thm_apx:general valuations}.

\begin{lemma}\label{vertex cover_apx}
An efficient truthful mechanism cannot provide an approximation ratio of $m^{1-\epsilon}$ w.r.t. $m^{\epsilon}$-local valuations, unless $NP\subseteq P/poly$.
\end{lemma}

\begin{proof}
Let $M$ be a mechanism as specified. Fix almost single mined, $m^{\epsilon}$-local valuations $v_{-i}$, $k$ and $0\le p\le 2$ for which  $|\cs(v_{-i},k,p)|\ge \frac{2^{\frac{m^{\frac{1}{2}\epsilon}}{2}}}{2m^6}$. Denote $\cs= \cs(v_{-i},k,p)$. By the Sauer-Shelah lemma, there is a set $U\subset [m]$ of size $m^{\frac{\epsilon}{4}}$ (for sufficiently large $m$) that is shattered by $\cs$.

We will reach a contradiction by showing that this entails an efficient (non-uniform) algorithm that detect whether a graph on the vertex set $U$ has a vertex-cover of size $t$.

Consider the valuation $v_i=v_i^1+v_i^2$ defined as follows: $v_i^1(S)=6m^2$ if $S\cap U$ is a vertex-cover and $0$ otherwise. While
\[
v^2_i(S)=
\begin{cases}
3 & (S\in \cs\text{ and }|S\cap U|=t)\text{ or }(p_{v_{-i}}(S)> p\text{ and }|S|>k)
\\
0 & \text{otherwise}
\end{cases}
\]
It is not hard to see that $v_i$ is nondecreasing, $m^\epsilon$-local, and by lemma \ref{lem:truthful_comp_apx} can be realized by a polynomial sized circuit. Also, by the taxation principle, it is not hard to check that if there is a cover $C$ of size $t$ of $U$, then the mechanism, operating on $v_1,\ldots,v_n$, will allocate the bidder $i$ a set $T$ such that $T\cap U$ is a cover of size $t$. This concludes the proof.
\end{proof}

\section{Single- and Multi-Minded Valuations}
Recall that single minded bidders are interested in a single bundle of items and are willing to pay a fixed price for the whole bundle (or for any superset), and pay nothing in any other case.
Though simply defined and easily represented, these valuation are complex enough such that even without requiring a mechanism to be incentive-compatible, one could not expect an approximation ratio better than $O(\sqrt{m})$. On the other hand, one of the few positive results in this scenario is a greedy algorithm \cite{lehmann2002truth} that obtains an  $O(\sqrt{m})$-approximation ratio while maintaining incentive compatibility.
One can verify that this greedy mechanism is not VCG-based, but still expect similar performance from VCG-based mechanisms. Our first theorem asserts that this is not the case.

\begin{theorem}\label{single minded thm}
For every $\epsilon>0$, no efficient VCG-based mechanism obtains an approximation ratio of $m^{1-\epsilon}$ w.r.t. single-minded valuations, unless $NP\subseteq P/poly$.
\end{theorem}

In cases where there are $d$ duplicate copies of each item, one might hope to boost up the poor performance of mechanisms, predicated by the above theorem. Indeed, a mechanism devised by Bartal et al. \cite{bartal2003incentive}, is both incentive compatible, have an approximation ratio of $dm^{\frac{1}{d-2}}$, and efficient whenever demand queries can be carried out efficiently (in particular, it is efficient for single minded bidders).

Once again, one might hope that the universal technique of VCG-based mechanisms can achieve similar (or even better) approximation ratios. A generalization of theorem \ref{single minded thm} shows that it is impossible.

\begin{theorem}\label{single minded duplicates_apx}
Assuming the UGC, for every constant $d$ and every $\epsilon>0$, no efficient VCG-based mechanism with $d$ duplicates have an approximation ratio of $m^{1-\epsilon}$ w.r.t. single-minded valuations, unless  $NP\subseteq P/poly$.
\end{theorem}

Single minded valuations, are thus computationally hard to deal with, while trivial from a communicational perspective - only $O(m)$ bits are required to encode each valuation. Considering not one, but several bundles of items, changes the picture:
\begin{theorem} \label{multi-minded thm_apx}
For every constant $d$ and $\epsilon>0$, if a VCG-based mechanism for combinatorial auctions with $d$-duplicates obtains an approximation ratio of $m^{1-\epsilon}$ w.r.t. multi-minded valuations, it must use exponential communication.
\end{theorem}

\subsection{Single-Minded Computational Bound -- Proof of Theorem \ref{single minded duplicates_apx}}
Suppose that $M$ is VCG-based, and obtains an approximation ratio of $m^{1-\epsilon}$  for  combinatorial auctions with $d$ duplicates w.r.t. single-minded valuations. Let $\mathcal{H}\subset P_d ([m],[n])$ denote its bank of allocations, and let $S\in P([m],[n])$ represent single minded bidders (bidder $i$ desires the bundle $S_i$). The approximation ratio implies that must be some $T\in \mathcal{H}$ such that
\[
n\leq m^{1-\epsilon}\cdot | \{i|S_i\subset T_i\}|~.
\]
Therefore, by theorem \ref{thm:shatter_single_minded}, there is a shattered pair $B\subset [n]$, $X\subset [m]$ of sizes $m^{\epsilon/4}$, $\Omega(m^{3\epsilon/4})$.

We will show computational hardness by a reduction from the \textbf{k-packing promise problem}. In this problem, the input consists of $r$ sets $S_1,....,S_r\subset U$ and
an integer $C>0$, and the goal is to distinguish between the following two options:
\begin{itemize}
	\item \textbf{Positive:} There are $C$ sets in $S_1,....,S_n$ that are pairwise disjoint.
	\item \textbf{Negative:} For every collection $\mathcal{F}$ of $C$ sets out of $S_1,....,S_n$, there is some $x\in U$ that is covered by $k$ sets of $\mathcal{F}$.
\end{itemize}
Observe that setting $k=2$ yields the famous "set packing" problem, which is known to be NP-hard. We shall make use of the fact that assuming UGC, this problem is NP-hard for every constant $k$. (The proof, based on \cite{bansal2010inapproximability}, can be found in the appendix).

The reduction is given by taking $k=d+1$, and using the shattered $X$ set as the universe $U$: Given $S_1,...,S_{m^{\epsilon/4}}\subset X$ and an integer $C>0$, suppose that each bidder $i\in B$ desires $S_i$, and all bidders not in $B$ have the zero valuation.
Now, in the ``yes" cases, since there is a collection of $C$ pairwise disjoint subsets from $S_1,...,S_{m^{\epsilon/4}}$, and since the pair $X,B$ is shattered and $M$ is MIR, the social welfare, $\sum_{i} v_i(S_i)$, will be at least $C$. On the other hand, in the ``NO" cases, it is impossible to satisfy $C$ bidders by {\em any} $d$-duplicate allocation, and therefore, the social welfare will be less that $C$.

\subsection{Multi-Minded Communication Bound -- Proof of Theorem \ref{multi-minded thm_apx}}

\begin{definition}
Let $\mathcal{F}=\{A^s\}_{s=1}^t$ be a collection of partitions such that each $A^s$ partition the universe $X$ into $k$ sets. We say that $\cf$ has the $k$-wise intersection property if for every choice of $k$ indices $1\leq s_1,...,s_k\leq t$, we have $\cap_{i=1}^k A_i ^{s_i} =\emptyset$ if and only if $s_1=\ldots=s_k$.
\end{definition}

The following lemma asserts that a large collection of partitions with such intersection property exists.

\begin{lemma}[essentially \cite{dobzinski2005optimal}]  \label{int. prop. col.}
There exists a collection of partitions $\mathcal{F}=\{A^s\}_{s=1}^t$ of the set $X$ into $k$ sets, such that $|\mathcal{F}|=t=e^{\frac{|X|}{k^{k+1}}}$, and $\cf$ has the $k$-wise intersection property holds.
\end{lemma}

\begin{proof}
We use a probabilistic construction: Choose each partition $A^s$ independently at random (that is, each element in $X$ will be placed in $A_j ^s$ with equal probability). Fix $k$ indices $1\leq s_1,...,s_k\leq t$ at least two of which are different. We have that
\[
\Pr[\cap_{i=1}^{k} A_i ^{s_i} =\emptyset]=\prod_{j=1}^{|X|} \Pr[j\notin\cap_{i=1}^{k} A_i ^{s_i}]\le\left(1-\frac{1}{k^k}\right)^{|X|} \leq e^{-\frac{|X|}{k^k}}
\]
As there are at most $t^{k}$ choices of indices, using the union bound we conclude that as long as $t^{k}<e^{\frac{|X|}{k^k}}$ such a collection of partitions exists.
\end{proof}

\begin{proof}
\textbf{(of Theorem \ref{multi-minded thm_apx})}.
Suppose that $M$ is VCG-based, and obtains an approximation ratio of $m^{1-\epsilon}$  for  combinatorial auctions with $d$ duplicates w.r.t. multi-minded valuations. Let $\mathcal{H}$ denote its bank of allocations. As single-minded valuations contained in this class, $\mathcal{H}$ shatters a pair of subsets $B\subset [n]$, $X\subset [m]$ of sizes $d+1$, $\Omega(m^{\epsilon/2})$. Applying Lemma \ref{int. prop. col.}, yields a collection of partitions $\mathcal{F}=\{D^s\}_{s=1}^t\subset P(X,B)$ of size $t=e^{m^{\frac{1}{2}\cdot\epsilon} / (d+1)^{d+2}}$ which satisfy the $(d+1)$-wise intersection property.

We now use this collection of partitions in order to show a reduction from the \textbf{approximate disjointness problem}: This problem deals with $r$ players, each of which holds a $t$-bit string that specify a subset $A_i\subseteq\{1,...,t\}$. The players are required to distinguish between the following two extreme cases:

\begin{itemize}
	\item \textbf{Positive:} $\cap_{i=1}^{r} A_i\neq\emptyset$.
	\item \textbf{Negative:} $A_i\cap A_j=\emptyset$ for all $1\leq i\neq j\leq r$.
\end{itemize}

The communication complexity of the approximate disjointness problem is known to be $\Omega\left(\frac{t}{r^4}\right)$ \cite{nisan2002communication}. The reduction is given as follows: Suppose that each of the bidders in $B$ holds a subset $A_i\subset \{1,\ldots,t\}$ which determines the multi-minded valuation
\[
v_i(T)=\begin{cases}
1, & \exists s\in A_i\ {s.t.\ } D_i^s\subseteq T \\
0, & else
\end{cases}
\]
Observe that if $\cap_{i=1}^{r} A_i\neq\emptyset$, indicated by some $s\in \cap_{i=1}^{r} A_i$, then the allocation $A^s$ satisfies all bidders in $B$, and since the pair $B,X$ is shattered, the social welfare will be $d+1$.

On the other hand, if $A_i\cap A_j=\emptyset$ for all $1\leq i\neq j\leq r$, then it is not hard to see that it is impossible to satisfy all the bidders in $B$ by {\em any} $d$-duplicate allocation, and therefore, the social welfare will be $\le k$.

We conclude that simulating these valuations (and zero valuations for all bidders not in $B$) solves the approximate disjointness problem, and the communication complexity of $M$ is therefore
\[
\Omega\left(\frac{t}{(d+1)^4}\right)=\Omega\left(\frac{e^{m^{\frac{1}{2}\cdot\epsilon} / (d+1)^{d+2}}}{(d+1)^4}\right)
\]
which proves the claim.
\end{proof}

\section{Submodular valuations}\label{sec:results-last}
In this section we sketch four lower bounds for submodular valuations and subclasses of submodular valuations.
This demonstrate the power and generality of the VC technique, as we show how to prove under this framework three known results and one that strengthen a previous result of \cite{dobzinski2007limitations}. Besides unification, the proofs sketched here are somewhat simpler. In particular, theorem \ref{thm:submodular_comm} that substantially simplifies (and strengthen) \cite{dobzinski2007limitations}

\subsection{Communication lower bounds on VGC mechanisms}

\begin{theorem} \label{thm:submodular_comm}
For every $\epsilon>0$, if a VCG-based mechanism have an approximation ratio of $m^{\frac{1}{3}-\epsilon}$ w.r.t. submodular valuations, it must use exponential communication.
\end{theorem}
This result strengthen a similar result due to \cite{dobzinski2007limitations}, who showed VCG mechanisms with approximation ratio of $m^{\frac{1}{6}-\epsilon}$ w.r.t. submodular valuations must use exponential communication. In addition, the proof sketched here is substantially simpler.
\begin{proof}(sketch)
Let $M$ be a VCG mechanism having an approximation ratio of $m^{\frac{1}{3}-\epsilon}$ w.r.t. submodular valuations. Let $\ch\subset P([m],[n])$ be its bank of allocations. Using theorem \ref{thm:shatter_additive_apx}, we can find a shattered pair $S,A$ where $|S|\ge m^{\epsilon}$ and $|A|=2$. Now, assume that all the players besides the players in $A:=\{1,2\}$ have the zero valuation, and the valuations of the bidders in $A$ depends only on the items in $S$.

The output of the mechanism induce an allocation of $S=A_1\dot\cup A_2$ that maximizes $v_1(A_1)+v_2(A_2)$. Therefore, the problem finding an allocation that maximizes the sum of two submodular functions can be solved using the mechanism. Since the solution of this problem requires exponential communication~\cite{nisan2006communication} the theorem follows.
\end{proof}

\subsection{Computational and value queries lower bounds}
\begin{theorem}\label{thm:sub_comp}
For any $\epsilon>0$, unless $NP\subseteq P/poly$, no efficient truthful mechanism has an approximation ratio of $m^{\frac{1}{2}-\epsilon}$ w.r.t. submodular valuations.
\end{theorem}
\begin{theorem}\label{thm:sub_val}
For any $\epsilon>0$ any truthful mechanism that communicates with the bidders only by value queries, and has an approximation ratio of $m^{\frac{1}{2}-\epsilon}$ w.r.t. submodular valuations, must use exponentially many value queries.
\end{theorem}
The first theorem was proved by \cite{dobzinski2012computational} under the weaker assumption that $NP\ne P$. The second theorem was proved by \cite{dobzinski2011impossibility}. The proof sketched here is close to the proofs in \cite{dobzinski2011impossibility,dobzinski2012computational}, but the step of embedding a hard problem given a ``large menu" is somewhat simplified.

\begin{proof} (of theorems \ref{thm:sub_comp} and \ref{thm:sub_val} -- sketch)
We will use the terminology of section \ref{sec:general_comp_apx}.
Assume $n=m^{\frac{1}{2}-\epsilon}$. Let $M$ be a truthful mechanism with approximation ratio $m^{\frac{1}{2}-\epsilon}$ w.r.t. submodular valuations. By lemma 2.3 in \cite{dobzinski2012computational}, there is $1\le k\le m$, $0\le p\le m$ and additive valuations $v_{-i}$ for which $|\cs(v_{-i},k,p)|\ge\frac{e^{m^{2\epsilon}}}{10m^8}$. Denote $\cs=\cs(v_{-i},k,p)$. By the Sauer-Shelah lemma, $\cs$ shatters a set $U\subset [m]$ of size $\ge m^\epsilon$ (for large enough $m$). Given collection $B$ of subsets of $S$, consider the following valuation
\[
v_i(S)=
\begin{cases}
t|S| & |S|< k
\\
tk & |S|=k,\;S\in \cs\text{ and }S\cap U\in B
\\
tk-\frac{1}{m^4} & |S|=k,\;S\in \cs\text{ and }S\cap U\not\in B
\\
t(k-2^{-k}) & |S|=k\text{ and }S\not\in \cs

\\
t(k-2^{-|S|}) & |S|> k\text{ and }p_{v_{-i}}(S)< p+\frac{1}{m^5}
\\
tk & |S|>k\text{ and }p_{v_{-i}}(S)\ge p+\frac{1}{m^5}
\end{cases}
\]
It is not hard to check that $v_i$ is nondecreasing, normalized and submodular for every $t\ge 1$. Also, by the taxation principle, for sufficiently large $t$, if the mechanism operates on $v_i,v_{-i}$ then it will allocate the bidder $i$ a sets $S\in \cs$ such that $S\cap U\in B$ (provided that such a set exists).

We now split the proof into two proofs (one for theorem \ref{thm:sub_comp} and the other for theorem \ref{thm:sub_val}).

{\it Theorem \ref{thm:sub_val}:} We note that given, $v_{-i},\cs$ and a {\em value query access} to the (characteristic function of the) collection $B$, it is possible to simulate a value query to $v_i$ by a single value query to $B$. Therefore, the mechanism can be used to solve the following problem: Given a value query access to $B$, determine whether $B$ is empty. It is clear that the query complexity of this last problem is $2^{|U|}-1$, and therefore, the query complexity of $M$ is at least $2^{|U|}-1\ge 2^{m^\epsilon}-1$.

{\it Theorem \ref{thm:sub_comp}:} We note that if membership in $B$ can be evaluated by a polynomial sized circuit, then, by lemma \ref{lem:truthful_comp_apx}, $v_i$ can be also realized by a polynomial sized circuit. Using this fact, it is not hard to show a (non-uniform) polynomial time algorithm for, say, MAXCUT. Given a graph on the vertices $U$ and a number $r$, we will define $B$ as the collection of all subsets $S\subset U$ such that the cut $(S,U\setminus S)$ consists of $\ge k$ edges. By running the mechanism with $v_i$ corresponding to this $B$, we will be able to efficiently identify whether the graph has a cut of size $\ge k$.
\end{proof}

\subsection{Computational lower bounds on capped additive valuations}

Our next result considers the class of valuations known as capped-additive (or budget-additive), where a bidder has an additive valuation over the items, together with a restriction on the maximal total price she is willing to pay.

\begin{definition}
A valuation function $v$ is said to be a \textbf{capped-additive} if there exist an additive valuation $a$ and a real value B, such that for every bundle
$S\subset [m]$, $v(S)= \min \{a(S),B\}$.
\end{definition}

\begin{theorem}
Let $\epsilon>0$. There is no polynomial-time VCG mechanism with an approximation ratio of $m^{0.5-\epsilon}$ w.r.t. capped additive valuations unless $NP\subset P/poly$.
\end{theorem}
This theorem have been proved in \cite{mossel2009vc,buchfuhrer2010inapproximability,buchfuhrer2009limits,dughmi2009amplified}, which first applied the notion of VC dimension in order to lower-bound the performance of mechanisms. The proof presented here somewhat simplifies the step of finding a shattered set.

\begin{proof} (sketch)
First, observe that capped additive valuations strictly contains additive valuations. Hence, the existence of such a mechanism $M$ together with theorem \ref{thm:shatter_additive_apx} implies that its bank of allocations  $\mathcal{H}$ shatters a pair $A\in \binom{[n]}{2}$, $S\subset [m]$ with, $|S|\ge m^{\epsilon/4}$. This allows us to embed SUBSET-SUM (see \cite{buchfuhrer2009limits}), and computaional hardness follows.
\end{proof}

\section{Hardness of the $k$-packing promise problem}\label{apx:k-packing}
Recall that the \textbf{k-packing promise problem} is the following. The input consists of $r$ sets $S_1,....,S_r\subset U$ and
an integer $C>0$, and the goal is to distinguish between the following two options:
\begin{itemize}
	\item \textbf{Positive:} There are $C$ sets in $S_1,....,S_n$ that are pairwise disjoint.
	\item \textbf{Negative:} For every collection $\mathcal{F}$ of $C$ sets out of $S_1,....,S_n$, there is some $x\in U$ that is covered by $k$ sets of $\mathcal{F}$.
\end{itemize}

\begin{theorem}
For every constant $k$, the $k$-packing promise problem is UGC-hard.
\end{theorem}

\begin{proof}
We shall use a reduction from the following promise problem: Given a $k$-uniform hypergraph $G = (V, E)$, one should distinguish between the following cases:

\begin{itemize}
\item \textbf{Yes case}: there exists $V'\subset V$, satisfying $|V'|\geq \frac{1}{2k} |V|$, such that no hyper-edge contains two or more vertices from $V'$.
\item \textbf{No case}: every set $V'$ of size $\ge \frac{1}{2k}|V|$ contains an edge.
\end{itemize}
This problem (actually, a more general one) was shown to be UGC-hard in \cite{bansal2010inapproximability}. The reduction to the $k$-packing promise problem is given by referring to the hyper-edges as the elements of the universe. For each vertex $i\in V$ define $S_i$ to be the set of all hyper-edges that contain $i$. Finally, let $C=\frac{1}{2k}|V|$.

It is straight-forward to verify that this defines a valid reduction to the $k$-packing promise problem.
\end{proof}

\section{Open Questions}

\noindent{\bf Inapproximability results for deterministic truthful mechanisms.} As Table~\ref{table} shows, important questions about the approximability of deterministic truthful mechanisms remain open, including the approximability of general truthful mechanisms for general and submodular valuations in the communication complexity model, and the approximability of general truthful mechanisms for capped-additive (and other classes of succinctly-described valuations) in the computational complexity model.

\vspace{0.05in}\noindent{\bf Inapproximability results for randomized truthful mechanisms.} We applied VC machinery to establish lower bounds for deterministic mechanisms. What about randomized mechanisms? Consider, e.g., the subclass of truthful mecahnisms that are maximal in \emph{distributional} range (MIDR~\cite{dobzinski2009power}), i.e., always output the best probability distribution over outcomes in a predetermined set of probability distributions (that need not contain all distributions). Can VC machinery be extended to prove lower bounds for MIDR (as with MIR)?

\vspace{0.05in}\noindent{\bf Optimization over partial domains.} SAT$_L$ is an example of an interesting, and largely unexplored, class of problems in complexity theory: optimization over ``partial domains'', i.e., when the space of possible solutions is (potentially) incomplete. Intuitively, any NP-hard problem remains hard so long as the space of solutions is ``large''. Our results here can, from this perspective, be regarded as providing the means for proving this in the context of set-packing problems. Partial-domain variants of other classical problems seem to require new machinery Consider, e.g., k-SET-COVER, in which the input is a collection of subsets $S_1,\ldots,S_n$ of and the goal is to determine whether $U$ can be covered by $k$ of these subsets. What about k-SET-COVER$_L$, in which the indices of the $k$ subsets must now belong to some large predetermined set $L\subseteq [n]^k$?

\vspace{0.05in}\noindent{\bf Tight bounds for the new dimension.} We observe that when $k=n$ the bound of Theorem~\ref{thm:sauer_gen} is tight for every choice of $d,m$. To see this, fix $m\ge d\ge 0$ and $n\ge k\ge 2$ and onsider the class $\ch\subset [n]^{[m]}$ consisting of all functions $f:[m]\to [n]$ for which $|\{x\mid f(x)\ge k\}|\le d$. Now, observe that $\Dim_{k}(\ch)=d$ and
\[
|\ch|=\sum_{i=0}^{d}\binom{m}{i}(k-1)^{n-i}(n+1-k)^{i}~.
\]

However, if $k<n$ the bounds derived from Theorem~\ref{thm:sauer_gen} are loose. Sharpening the upper/lower bounds is left for future research.

\end{document}